\documentclass[journal]{IEEEtran}
\ifCLASSINFOpdf

\else

\fi
\usepackage{amsmath}
\usepackage{makeidx}  
\usepackage{algorithm}
\usepackage{algorithmic}
\usepackage{graphicx}
\usepackage{subfigure}
\usepackage{epstopdf}
\usepackage{bm}
\usepackage{cite}
\usepackage{stfloats}

\usepackage{amssymb}
\usepackage{graphicx}

\usepackage{url}

\usepackage{tabularx,booktabs}
\newcolumntype{C}{>{\centering\arraybackslash}X} 
\usepackage{lipsum}

\usepackage{makecell} 

\usepackage{graphicx}
\usepackage{color}
\newtheorem{thm}{Theorem}
\newtheorem{rem}{Remark}
\newtheorem{lem}{Lemma}
\newtheorem{pos}{Proposition}
\newtheorem{proof}{proof}

\begin{document}

\title{Active Sensing-assisted UAV Communications with Jittering: Framework and Performance Analysis}

\author{Guangji Chen,
        Long Shi,
        Qingqing Wu,
        Qiaoyan Peng,
        and Caihong Kai
        \thanks{Guangji Chen and Long Shi are with Nanjing University of Science and Technology, Nanjing 210094, China (email:
                guangjichen@njust.edu.cn; longshi@njust.edu.cn). Qingqing Wu is with Shanghai Jiao Tong University, 200240, China (e-mail: qingqingwu@sjtu.edu.cn). Qiaoyan Peng is with University of Macau, Macao 999078, China (email: yc27464@umac.mo). Caihong Kai is with  Hefei University of Technology, Hefei 230601, China (e-mail: chkai@hfut.edu.cn).
                }}

\maketitle
\vspace{-16pt}
\begin{abstract}
Providing reliable communication for unmanned aerial vehicles (UAVs) via existing cellular networks is crucial for enabling the rapid growth of the low-altitude economy. However, UAV jittering significantly degrades communication quality due to induced beam misalignment. Inspired by recent advances in integrated sensing and communication, we propose a novel two-stage active sensing-assisted communication framework tailored for ground-to-UAV links with jittering. Specifically, two schemes are conceived to leverage sensing for enhancing communication performance, namely the communication-oriented scheme and the sensing-oriented scheme. For the sensing-oriented scheme, deterministic signals are employed in the first stage to facilitate angle-of-arrival (AoA) acquisition at the UAV side, followed by pure communication service in the second stage by using the estimated AoA. In contrast, the communication-oriented scheme employs Gaussian information-bearing signals throughout both stages, with AoA estimation relying on Gaussian random signals. For both schemes, we provide maximum likelihood estimators for AoA, along with analytical results characterizing the Cram\'er-Rao bound. To capture the performance limit, closed-form expressions for the achievable rates of the two schemes are derived, unveiling a fundamental tradeoff between sensing and communication quality across the two stages by tuning the time allocated to the first stage. The optimal time allocation that maximizes the overall rate is obtained in semi-closed-form. Based on these results, we unveil a sufficient condition under which the communication-oriented scheme outperforms the sensing-oriented scheme, which admits an interesting threshold-based structure. Asymptotic analysis demonstrates that the performance loss of the proposed schemes relative to the jitter-free upper bound approaches zero in the high transmit power regime. Finally, simulation results verify the theoretical findings and the superiority of the proposed schemes.
\end{abstract}

\begin{IEEEkeywords}
beamforming, integrated sensing and communication, jittering effects, UAV.
\end{IEEEkeywords}

\IEEEpeerreviewmaketitle

\vspace{-8pt}
\section{Introduction}

\vspace{-2pt}
The low-altitude economy (LAE) has attracted significant attention for its potential to achieve unprecedented efficiency in the utilization of vertical airspace below 1,000 meters, thereby enabling both existing and emerging applications such as low-altitude transportation, precision agriculture, and remote sensing \cite{jiang20236g}. As a core component of LAE, the production of unmanned aerial vehicles (UAVs) is predicted to experience sustained rapid growth over the next decade, with the global market value expected to exceed 70 billion dollars by 2030 \cite{khawaja2025survey}. From a wireless communication perspective, UAVs offer several unique advantages, including flexible deployment, three-dimensional (3D) maneuverability, and the ability to establish line-of-sight (LoS) communication links \cite{zeng2019accessing}. To lay the foundation for LAE applications, it is crucial to provide reliable communication services to UAV users via existing cellular networks. To this end, employing large antenna arrays at transceivers is appealing for improving ground-to-air communication quality through high beamforming gain \cite{yuan2022joint}.

Different from conventional beamforming techniques in terrestrial communication systems, beamforming design in a low-altitude network faces a unique challenge arising from UAV jittering. In practice, external disturbances, such as air turbulence and wind gusts, may cause significant fluctuations in the yaw, pitch, and roll angles of a UAV platform, thereby leading to high-frequency variations in the channel-related angle of arrival (AoA) or angle of departure (AoD) between transceivers \cite{arain2014real}. Due to the uncertainty in AoA/AoDs induced by UAV jittering,  beam misalignment severely degrades the beamforming gain, resulting in substantial performance loss \cite{yang2022impact,wu2020secrecy,dabiri2020analytical,wang2021jittering}. To address the issue caused by UAV jittering, growing research attention has been devoted to developing robust beamforming schemes for UAV communication systems against jittering effects \cite{liu2023deployment,lee2024robust,tang2024sensing,liu2022deployment,ouyang2025robust,chen2024adaptive}. The key idea of these works is to exploit beam broadening techniques to enlarge the beam coverage area so that it encompasses the angular range of potential AoA/AoD errors. For example, the authors of work \cite{chen2024adaptive} proposed an effective robust beamforming scheme by considering the asymmetric effect of jitter on the angular domain information in a UAV communication system. Although these approaches alleviate the unfavorable impact of UAV jittering on system performance, there exists an inherent tradeoff between the beamforming gain and the coverage beamwidth,  which inevitably leads to a loss of beamforming gain, especially in the regime of large angular information errors \cite{chen2023static}.

Besides the robust beamforming designs, the recent advances in integrated sensing and communication (ISAC) offer new opportunities to meet the requirements of high-throughput, low-latency, and ultra-accurate sensing in low-altitude networks \cite{liu2022integrated,zhang2026integrated,wu2026Intelligent,song2025overview}. The fundamental concept of ISAC is to co-design the dual functionalities of sensing and communication on the same wireless platform (e.g., the same frequency band and hardware) \cite{liu2022integrated}. On the one hand, ISAC can achieve integration gains to enable a flexible tradeoff between sensing and communication performance via dedicated waveform design, beamforming design, antenna deployment, and resource allocation \cite{meng2024network,chen2026multi,ren2023fundamental}. On the other hand, the synergy between communication and sensing unlocks the potential for coordination gains, e.g., communication-assisted sensing and sensing-assisted communication \cite{liu2022integrated}. Among these, sensing-assisted communication aims to estimate  channel-related parameters (e.g, location, angle, and distance) via ISAC signals, which is helpful for enhancing communication performance by reducing the associated channel acquisition overhead \cite{liu2020radar}. For example, wireless sensing has been applied for tracking and predictive beamforming in vehicle-to-infrastructure systems by capturing the channel variation in high-mobility environments \cite{meng2023sensing,yuan2020bayesian}.

Owing to the wireless sensing capability of ISAC, integrating it into UAV communications opens new avenues for addressing the challenges induced by time-varying channels \cite{song2025overview}. Generally, there are two typical radio-based paradigms for UAV ISAC, namely UAV for ISAC and ISAC for UAV \cite{mu2023uav}. The former employs dedicated UAVs as aerial platforms to provide ISAC services by fully harnessing the 3D mobility of UAVs \cite{meng2023uav}, while the latter utilizes cellular networks to enable both communications and surveillance for UAV users/targets \cite{song2025overview}. Benefiting from the controllable 3D mobility of UAVs, the joint optimization of UAV trajectory and beamforming was widely investigated to achieve a flexible tradeoff between sensing and communication \cite{chen2026rotatable,jing2024isac,lyu2022joint}, as well as to improve communication performance through sensing \cite{pang2024dynamic,jiang2024energy}. In addition, the densely deployed terrestrial wireless network infrastructure facilitates seamless wireless coverage for UAVs. In the context of ISAC for UAV, the authors of \cite{cheng2025networked} studied a cooperative ISAC system, where multiple base stations (BSs) perform downlink communications with UAVs while sensing a given target area. Moreover, in \cite{song2026integrated}, the performance boundary of UAV detection probability and achievable rate was characterized via beamforming design, which reveals the fundamental tradeoff between sensing and communication. Beyond the objective of achieving a flexible tradeoff between sensing and communication, prior works \cite{xu2026deep,jiang2025low} devoted to the topic of sensing-assisted beamforming for UAV communication performance enhancement, where the time-varying BS-UAV channels are acquired based on received echo sensing signals.

Existing studies on sensing-assisted communication for both the terrestrial vehicles \cite{pang2024dynamic} and  low-altitude UAVs \cite{xu2026deep,jiang2025low} mainly focused on tracking the locations of mobile targets. In a case of line-of-sight dominant channel with known transceiver orientations, the end-to-end channel is purely dependent on the relative positions of transceivers, which implies that tracking the location of the target is sufficient to reconstruct its wireless channel. However, for the scenario of cellular-connected UAVs, UAV jittering inevitably introduces uncertainties in the UAV attitude. Consider a typical setup where a ground BS provides downlink data transmission to a cellular connected UAV user. The wireless channel in this setup is related to both the transceiver positions and the UAV attitude. Notice that the UAV attitude is highly affected by jittering effects of UAV \cite{yang2022impact,wu2020secrecy,dabiri2020analytical,wang2021jittering}. As discussed in a seminal work \cite{wang2021jittering}, UAV jittering primarily  introduces variations in the AoA at the UAV side rather than the AoD at the BS side, which consequently causes beam misalignment at the receiver. It is worth noting that the widely investigated passive sensing-assisted communication paradigm relying on echo signals \cite{pang2024dynamic,xu2026deep,jiang2025low} is not applicable to the scenario of UAV communication with jittering since the echo signals contain no information of the UAV-specific attitude. To fully unleash the potential of high beamforming gain provided by large antenna arrays in the context of UAV communication with jittering, a sensing-assisted communication paradigm tailored for this particular scenario is required, which thus motivates our work.

\begin{figure}[!t]
\setlength{\abovecaptionskip}{-5pt}
\setlength{\belowcaptionskip}{-5pt}
\centering
\includegraphics[width= 0.5\textwidth]{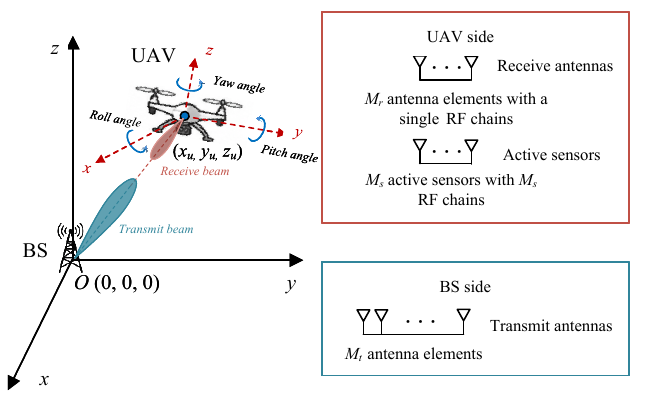}
\DeclareGraphicsExtensions.
\caption{Illustration of the system model with UAV jittering.}
\label{model1}
\vspace{-12pt}
\end{figure}

In this paper, we propose a novel active sensing-assisted communication framework tailored for ground-to UAV link with jittering, shown in Fig.\ref{model1}. By deploying active sensors on the UAV side, the UAV can proactively sense the channel-related angle information from the signals received by these active sensors. This fundamentally differs from existing works relying on echo signals for passive sensing. Specifically, a time slot with constant UAV attitude is divided into two stages, where the first stage performs AoA estimation and the second stage provides pure communication services using the estimated AoA. Based on different signal designs, we propose two schemes to leverage sensing for improving communication performance, namely the communication-oriented scheme and the sensing-oriented scheme. The former uses the Gaussian information-bearing signals for sensing in the first stage, whereas the latter employs deterministic signals. First, for both schemes, the sensing accuracy in the first stage and the communication rate in the second stage needs to be balanced to maximize the overall achievable rate, which naturally leads to a fundamental tradeoff between sensing and communication. Second, each scheme has its own merits. Compared to the sensing-oriented scheme, the communication-oriented scheme achieves lower sensing accuracy due to the unavailability of signal samples at the receiver, but it provides an additional data rate in the first stage thanks to the use of Gaussian random signals. Hence, which scheme is more beneficial for maximizing the overall communication rate remains an open problem.  To shed light on the above considerations, the main contributions of this paper are summarized as follows:
\begin{itemize}
  \item  For both schemes, we provide maximum likelihood estimators (MLEs) for the AoA, along with analytical results characterizing the corresponding Cram\'er-Rao bound (CRB). Furthermore, we derive closed-form expressions for the overall achievable rates of both schemes in terms of system parameters, revealing a flexible trade-off between sensing and communication quality across the two stages by adjusting the time allocated to the first stage.
  \item To capture the maximum achievable system performance, we derive the optimal time allocation that maximizes the overall rate, providing semi-closed-form expressions that lay the foundation for the subsequent performance comparison. Based on these results, we unveil a sufficient condition under which the communication-oriented scheme outperforms the sensing-oriented scheme, which exhibits an interesting threshold-based structure by showing that the rate of the communication-oriented scheme in the first stage should be larger than a threshold.
  \item To understand the gap between the proposed schemes and the system performance limit, we first reveal that the optimal time allocated to the first stage scales on the order of the square root of the time slot length, implying the merit of low sensing overhead for the proposed schemes. Then, we further analytically prove that the performance loss of the proposed schemes relative to the jitter-free upper bound approaches zero in the high transmit power regime.
  \item Numerical results validate the theoretical findings and demonstrate the superiority of the proposed schemes over various benchmarks. Under a given total cost, the optimal combination of the number of active sensors and receive antennas at the UAV is explored via simulations, thereby providing practical insights for cost-constrained design.
\end{itemize}

\emph{Notations:} Boldface upper-case and lower-case  letter denote matrix and   vector, respectively.  ${\mathbb C}^ {d_1\times d_2}$ stands for the set of  complex $d_1\times d_2$  matrices. For a complex-valued vector $\bf x$, ${\left\| {\bf x} \right\|}$ represents the  Euclidean norm of $\bf x$, ${\rm arg}({\bf x})$ denotes  the phase of   $\bf x$, and ${\rm diag}(\bf x) $ denotes a diagonal matrix whose main diagonal elements are extracted from vector $\bf x$.
For a vector $\bf x$, ${\bf x}^*$ and  ${\bf x}^H$  stand for  its conjugate and  conjugate transpose respectively. A circularly symmetric complex Gaussian random variable $x$ with mean $ \mu$ and variance  $ \sigma^2$ is denoted by ${x} \sim {\cal CN}\left( {{{\mu }},{{\sigma^2 }}} \right)$.  $x \sim {\cal O}\left( y \right)$ indicates that $\lim \frac{x}{y} = c$, where $c$ is a constant. ${\mathop{\rm E}\nolimits} \left\{ X \right\}$ stands for an expectation over a random variable $X$.

\section{System Model }
As illustrated in Fig. \ref{model1}, we consider a ground-to-UAV wireless communication system, where a ground BS provides downlink data services for a cellular-connected UAV \footnote{The proposed framework can be readily extended to a general multi-user case by incorporating multiple access schemes \cite{chen2023active}.}. In a 3D Cartesian coordinate system, the locations of the BS and the UAV are denoted by ${{\bf{q}}_{\rm{B}}} = {\left[ {0,0,0} \right]^T}$ and ${{\bf{q}}_{\rm{U}}} = {\left[ {{x_{\rm{U}}},{y_{\rm{U}}},{z_{\rm{U}}}} \right]^T}$, respectively. We consider a ground BS equipped with a uniform linear array (ULA) with ${M_t}$ antenna elements along the $x$-axis. Similarly, the UAV is equipped with an ${M_r}$-element ULA. To reduce hardware cost, all ${M_r}$ antenna elements at the UAV are connected to only a single radio-frequency (RF) chain. Consequently, analog receive beamforming is carried out at the UAV to enhance the ground-to-UAV transmission.

Different from the ground BS, the UAV generally hovers in the sky without being tied to some stable infrastructures. Therefore, random wind gusts may induce UAV jittering. Such jittering effects significantly influence the UAV's orientation, which is characterized by its attitude ${\Theta _u} = \left( {{\alpha _u},{\beta _u},{\gamma _u}} \right)$, where ${{\alpha _u}}$, ${\beta _u}$, and ${{\gamma _u}}$ denote the yaw, pitch, and roll angles, respectively. The wireless channel from the BS to the UAV depends on both their relative positions and the UAV's attitude ${\Theta _u}$. In the considered system, the channel-related parameters, i.e., ${{\bf{q}}_{\rm{U}}}$ and ${\Theta _u}$, are functions of time $l \in \left[ {0,L} \right]$, with $L$ denoting the maximum duration during which the UAV remains within the coverage of the ground BS. The entire time period $L$ is divided into $N$ time slots, each of length $T$, indexed by $n \in {\cal N} \buildrel \Delta \over = \left\{ {1, \ldots ,N} \right\}$ such that $L = NT$. By setting $T$ sufficiently small, it is practical to assume that $\left\{ {{{\bf{q}}_U},{\Theta _u}} \right\}$ remains constant within each time slot. To sense the real-time channel-related parameters in each time slot, ${M_s}$ active sensors arranged in a ULA are deployed at the UAV, parallel to the ULA of receive antennas. Moreover, each active sensor is connected to an individual RF chain to facilitate baseband signal processing for parameter estimation. Given the high hardware cost of RF chains, we consider a setup with $M_s\ll M_r$.

\subsection{Channel Model with UAV Jittering}
In this subsection, we elaborate the channel model for the considered ground-to-UAV communication setup. Let ${{\bf{H}}_{\rm{c}}} \in {\mathbb{C}^{{M_r} \times {M_t}}}$ denote the baseband wireless channel from the BS to the receive antennas at the UAV. In mmWave bands, the channel is known to be dominated by the line-of-sight (LoS) component. Following the geometry-based channel model under LoS-dominant conditions, ${{\bf{H}}_{\rm{c}}} \in {\mathbb{C}^{{M_r} \times {M_t}}}$ can be expressed as
\begin{align}\label{channel_Hc}
{{\bf{H}}_{\rm{c}}} = {\beta _c}{{\bf{b}}_{\rm{c}}}\left( {{\mu _{\rm{A}}}} \right){{\bf{a}}^H}\left( {{\mu _{\rm{D}}}} \right),
\end{align}
where ${\beta _c} = \frac{\lambda }{{4\pi {d_{{\rm{BU}}}}}}{e^{ - j\frac{{2\pi {d_{{\rm{BU}}}}}}{\lambda }}}$ with $\lambda$ being the wavelength and ${d_{{\rm{BU}}}} = \left\| {{{\bf{q}}_{\rm{U}}} - {{\bf{q}}_{\rm{B}}}} \right\|$ denoting the distance between the BS and the UAV. The parameters ${{\mu _{\rm{D}}}}$ and ${\mu _{\rm{A}}}$ denote the AoD and AoA, respectively, in the cosine angular domain. The array response vectors ${\bf{a}}\left( {{\mu _{\rm{D}}},{\nu _{\rm{D}}}} \right)$ and ${{\bf{b}}_{\rm{c}}}\left( {{\mu _{\rm{A}}},{\nu _{\rm{A}}}} \right)$ are given by
\begin{align}\label{array_response}
{\bf{a}}\left( {{\mu _{\rm{D}}}} \right) = {\bf{u}}\left( {{\mu _{\rm{D}}},{M_t}} \right),{{\bf{b}}_{\rm{c}}}\left( {{\mu _{\rm{A}}}} \right) = {\bf{u}}\left( {{\mu _{\rm{A}}},{M_r}} \right),
\end{align}
where ${\bf{u}}\left( {x,M} \right) = {\left[ {{e^{ - j\pi \frac{{M - 1}}{2}x}}, \ldots ,{e^{j\pi \frac{{M - 1}}{2}x}}} \right]^T}$ is the steering vector for an $M$-element ULA with direction $x$ in the cosine angular domain.

Then, we demonstrate the relationship of the AoA/AoD $\left\{ {{\mu _{\rm{A}}},{\mu _{\rm{D}}}} \right\}$ with respect to $\left\{ {{{\bf{q}}_{\rm{B}}},{{\bf{q}}_{\rm{U}}},{\Theta _u}} \right\}$. Since the ULA at the BS is oriented along the $x$-axis, the corresponding AoD ${\mu _{\rm{D}}}$ is given by
\begin{align}\label{AoD_information}
{\mu _{\rm{D}}} = \left[ {1,0,0} \right]\frac{{{{\bf{q}}_{\rm{B}}} - {{\bf{q}}_{\rm{U}}}}}{{\left\| {{{\bf{q}}_{\rm{B}}} - {{\bf{q}}_{\rm{U}}}} \right\|}} = \frac{{ - {x_{\rm{U}}}}}{{\sqrt {x_{\rm{U}}^2 + y_{\rm{U}}^2 + z_{\rm{U}}^2} }}.
\end{align}
It can be observed from \eqref{AoD_information} that ${\mu _{\rm{D}}}$ depends only on the relative positions $\left\{ {{{\bf{q}}_{\rm{B}}},{{\bf{q}}_{\rm{U}}}} \right\}$ and is independent of the UAV attitude ${{\Theta _u}}$. At the UAV side, the orientation of the ULA is highly influenced by ${{\Theta _u}}$. Under the given ${\Theta _u} = \left( {{\alpha _u},{\beta _u},{\gamma _u}} \right)$, the rotation matrix is defined as ${{\bf{R}}_{\rm{U}}} = {{\bf{R}}_{{\rm{yaw}}}}{{\bf{R}}_{{\rm{pitch}}}}{{\bf{R}}_{{\rm{roll}}}}$, where ${{\bf{R}}_{{\rm{yaw}}}}$, ${{\bf{R}}_{{\rm{pitch}}}}$, and ${{\bf{R}}_{{\rm{roll}}}}$ are given by
\begin{align}\label{rotation_matrix}
&{{\bf{R}}_{{\rm{yaw}}}} = \left[ {\begin{array}{*{20}{c}}
{\cos {\alpha _u}}&{ - \sin {\alpha _u}}&0\\
{\sin {\alpha _u}}&{\cos {\alpha _u}}&0\\
0&0&1
\end{array}} \right],\nonumber\\
&{{\bf{R}}_{{\rm{pitch}}}} = \left[ {\begin{array}{*{20}{c}}
{\cos {\beta _u}}&0&{ - \sin {\beta _u}}\\
0&1&0\\
{\sin {\beta _u}}&0&{\cos {\beta _u}}
\end{array}} \right],\nonumber\\
&{{\bf{R}}_{{\rm{roll}}}} = \left[ {\begin{array}{*{20}{c}}
1&0&0\\
0&{\cos {\gamma _u}}&{ - \sin {\gamma _u}}\\
0&{\sin {\gamma _u}}&{\cos {\gamma _u}}
\end{array}} \right].
\end{align}

Note that the initial orientation of the ULA at the UAV without jittering is along the $x$-axis when ${\Theta _u} = \left( {0,0,0} \right)$. Under the attitude ${\Theta _u} = \left( {{\alpha _u},{\beta _u},{\gamma _u}} \right)$ with jittering, the orientation of the ULA at UAV becomes ${{\bf{R}}_{\rm{U}}}{\left[ {1,0,0} \right]^T}$.
Hence, the corresponding AoA ${\mu _{\rm{A}}}$ can be expressed as
\begin{align}\label{AoA_information}
{\mu _{\rm{A}}} = \left[ {1,0,0} \right]{\bf{R}}_{\rm{U}}^T\frac{{{{\bf{q}}_{\rm{B}}} - {{\bf{q}}_{\rm{U}}}}}{{\left\| {{{\bf{q}}_{\rm{B}}} - {{\bf{q}}_{\rm{U}}}} \right\|}}.
\end{align}
Different from the AoD ${\mu _{\rm{D}}}$, the AoA ${\mu _{\rm{A}}}$ in \eqref{AoA_information} depends on both the location information $\left\{ {{{\bf{q}}_{\rm{B}}},{{\bf{q}}_{\rm{U}}}} \right\}$ and the UAV attitude ${\Theta _u} = \left( {{\alpha _u},{\beta _u},{\gamma _u}} \right)$.

Owing to UAV jittering, the time-varying attitude of the UAV in the current time slot $n$ is given by
\begin{align}\label{attitude_updation}
{\Theta _u} = {\bar \Theta _u} + \Delta {\Theta _u},
\end{align}
where ${{\bar \Theta }_u} = \left( {{{\bar \alpha }_u},{{\bar \beta }_u},{{\bar \gamma }_u}} \right)$ denotes the UAV attitude in the previous time slot $n-1$ and $\Delta {\Theta _u} = \left( {\Delta {\alpha _u},\Delta {\beta _u},\Delta {\gamma _u}} \right)$ represents the random fluctuation due to jittering around the previous attitude. According to existing works \cite{yang2022impact,wu2020secrecy,dabiri2020analytical,wang2021jittering}, ${\Delta {\alpha _u}}$, ${\Delta {\beta _u}}$, and ${\Delta {\gamma _u}}$ can be modeled as independent Gaussian distributed random variables, i.e., $\Delta {\alpha _u} \sim {\cal N}\left( {0,\sigma _\alpha ^2} \right)$, $\Delta {\beta _u} \sim {\cal N}\left( {0,\sigma _\beta ^2} \right)$, and $\Delta {\gamma _u} \sim {\cal N}\left( {0,\sigma _\gamma ^2} \right)$ with ${\sigma _\alpha ^2}$, ${\sigma _\beta ^2}$, and ${\sigma _\gamma ^2}$ denoting the variances of
the fluctuations in  the respective attitude angles. Using the first-order Taylor series expansion, the AoA ${\mu _{\rm{A}}}$ in the current time slot $n$ is approximated as
\begin{align}\label{AoA_approximation}
{\mu _{\rm{A}}} = {{\bar \mu }_{\rm{A}}} + {\bf{J}}\left( {{{\bar \alpha }_u},{{\bar \beta }_u},{{\bar \gamma }_u}} \right){\left[ {\Delta {\alpha _u},\Delta {\beta _u},\Delta {\gamma _u}} \right]^T},
\end{align}
where ${{\bar \mu }_{\rm{A}}}$ is the corresponding AoA in the previous time slot $n - 1$ under ${{\bar \Theta }_u}$ and ${\bf{J}}\left( {{{\bar \alpha }_u},{{\bar \beta }_u},{{\bar \gamma }_u}} \right)$ is the Jacobian matrix defined as
\begin{align}\label{Jacobian_matrix}
{\bf{J}}\left( {{\alpha _u},{\beta _u},{\gamma _u}} \right) = \left[ {\frac{{\partial {\mu _{\rm{A}}}}}{{\partial {\alpha _u}}},\frac{{\partial {\mu _{\rm{A}}}}}{{\partial {\beta _u}}},\frac{{\partial {\mu _{\rm{A}}}}}{{\partial {\gamma _u}}}} \right].
\end{align}
Let $\Delta {\mu _{\rm{A}}} = {\bf{J}}\left( {{{\bar \alpha }_u},{{\bar \beta }_u},{{\bar \gamma }_u}} \right){\left[ {\Delta {\alpha _u},\Delta {\beta _u},\Delta {\gamma _u}} \right]^T}$ denote the random fluctuation of AoA caused by UAV jittering. Then, $\Delta {\mu _{\rm{A}}}$ is Gaussian distributed as ${\cal N}\left( {0,\sigma _\mu ^2} \right)$, with $\sigma _\mu ^2 = {\bf{J}}\left( {{{\bar \alpha }_u},{{\bar \beta }_u},{{\bar \gamma }_u}} \right){\rm{diag}}\left( {\sigma _\alpha ^2,\sigma _\beta ^2,\sigma _\gamma ^2} \right){{\bf{J}}^T}\left( {{{\bar \alpha }_u},{{\bar \beta }_u},{{\bar \gamma }_u}} \right)$. Consequently, the instantaneous AoA ${\mu _{\rm{A}}}$ in each time slot $n$ fluctuates randomly as ${\mu _{\rm{A}}} = {{\bar \mu }_{\rm{A}}} + \Delta {\mu _{\rm{A}}}$.

Similarly, the wireless channel from the ground BS to the ULA of active sensors deployed at the UAV is denoted by ${{\bf{H}}_{\rm{s}}} \in {\mathbb{C}^{^{M_{\rm{s}} \times {M_t}}}}$. Since this ULA is arranged to be parallel to the orientation of the ULA of
receive antennas, ${{\bf{H}}_{\rm{s}}}$ is given by
\begin{align}\label{channel_Hs}
{{\bf{H}}_{\rm{s}}} = \beta _{\rm{s}}^x{{\bf{b}}_{\rm{s}}}\left( {{\mu _{\rm{A}}}} \right){{\bf{a}}^H}\left( {{\mu _{\rm{D}}}} \right).
\end{align}
where ${{\bf{b}}_{\rm{s}}}\left( {{\mu _{\rm{A}}}} \right) = {\bf{u}}\left( {{\mu _{\rm{A}}},{M_{\rm{s}}}} \right)$ and $\beta _{\rm{s}}^x = {\beta _c}{e^{j\phi _{\rm{s}}^x}}$  with ${\phi _{\rm{s}}^x}$  being the additional phase difference relative to ${\beta _c}$ caused by the slight displacement between the center of the active sensor array and that of the receive antenna array at the UAV.

\subsection{Signal Model}
Under the given channel model, the communication signal received at the UAV's receive antennas is given by
\begin{align}\label{received signal_model}
{y_{\rm{c}}} = {{\bf{v}}^H}{{\bf{H}}_{\rm{c}}}{\bf{w}}{s_{\rm{c}}} + {{\bf{v}}^H}{{\bf{n}}_{\rm{c}}},
\end{align}
where ${s_{\rm{c}}}$ is the transmitted symbol satisfying ${\mathop{\rm E}\nolimits} \left[ {{{\left| {{s_{\rm{c}}}} \right|}^2}} \right] \le {P_{\rm{t}}}$ with ${P_{\rm{t}}}$ denoting the maximum transmit power at the BS, ${\bf{w}} \in {\mathbb{C}^{{M_t} \times {1}}}$ is the transmit beamforming vector at the BS satisfying ${\left\| {\bf{w}} \right\|_2^2} = 1$, ${\bf{v}} \in {\mathbb{C}^{{M_r} \times {1}}}$ is the analog receive beamforming vector at the UAV satisfying $\left| {{{\left[ {\bf{v}} \right]}_m}} \right| = 1/\sqrt {{M_r}} ,\forall m$, and ${{\bf{n}}_{\rm{c}}}$ is the additive white Gaussian noise (AWGN) satisfying ${{\bf{n}}_{\rm{c}}} \sim {\cal CN}\left( {{\bf{0}},{\sigma ^2}{{\bf{I}}_{{M_r}}}} \right)$. Based on \eqref{received signal_model}, the received SNR of the communication signal can be expressed as
\begin{align}\label{SNR_c}
{\rho _c} &= \frac{{{P_t}{{\left| {{{\bf{v}}^H}{{\bf{H}}_{\rm{c}}}{\bf{w}}} \right|}^2}}}{{{\sigma ^2}}}\nonumber\\
&  = \frac{{{P_t}{{\left| {{\beta _c}} \right|}^2}{{\left| {{{\bf{v}}^H}{{\bf{b}}_{\rm{c}}}\left( {{\mu _{\rm{A}}}} \right)} \right|}^2}{{\left| {{{\bf{a}}^H}\left( {{\mu _{\rm{D}}}} \right){\bf{w}}} \right|}^2}}}{{{\sigma ^2}}}.
\end{align}

Notice that achieving high-rate transmissions relies heavily on the joint optimization of the transmit and receive beamforming vectors $\left\{ {{\bf{w}},{\bf{v}}} \right\}$. From \eqref{SNR_c}, it is observed that maximizing the received SNR is equivalent to the following two independent optimization problems:
\begin{align}\label{opt_transmit_bf}
\mathop {\max }\limits_{\bf{w}} {\left| {{{\bf{a}}^H}\left( {{\mu _{\rm{D}}}} \right){\bf{w}}} \right|^2} ~~~{\rm{s.t.}}~{\left\| {\bf{w}} \right\|_2^2} = 1.
\end{align}
\begin{align}\label{opt_receive_bf}
\mathop {\max }\limits_{\bf{v}} {\left| {{{\bf{v}}^H}{{\bf{b}}_{\rm{c}}}\left( {{\mu _{\rm{A}}}} \right)} \right|^2} ~~~{\rm{s.t.}}~\left| {{{\left[ {\bf{v}} \right]}_m}} \right| = \frac{1}{{\sqrt {{M_r}} }},\forall m,
\end{align}
Thus, the optimal ${\bf{w}}$ depends only on ${\mu _{\rm{D}}}$, while the optimal ${\bf{v}}$ depends only on ${\mu _{\rm{A}}}$. Given accurate knowledge of ${\mu _{\rm{A}}}$ and ${\mu _{\rm{D}}}$, it is not difficult to show that the optimal solutions of problems \eqref{opt_transmit_bf} and \eqref{opt_receive_bf} are given by
\begin{align}\label{opt_tx_rx_bf}
{{\bf{w}}^*} = \frac{{{\bf{a}}\left( {{\mu _{\rm{D}}}} \right)}}{{\left\| {{\bf{a}}\left( {{\mu _{\rm{D}}}} \right)} \right\|}},{{\bf{v}}^*} = \frac{{{{\bf{b}}_{\rm{c}}}\left( {{\mu _{\rm{A}}}} \right)}}{{\left\| {{{\bf{b}}_{\rm{c}}}\left( {{\mu _{\rm{A}}}} \right)} \right\|}}.
\end{align}

As shown in an existing work \cite{wang2021jittering}, the 3D position of the UAV can be accurately obtained using global positioning system (GPS) and a barometer, enabling the BS to form optimal transmit beamforming towards the UAV. In contrast, accurate AoA information ${\mu _{\rm{A}}}$ is unavailable due to the randomness of UAV jittering. Directly forming the receive beamforming based on the AoA ${{\bar \mu }_{\rm{A}}}$ of the previous time slot can lead to a significant reduction in the receive beamforming gain ${\left| {{{\bf{v}}^H}{{\bf{b}}_{\rm{c}}}\left( {{\mu _{\rm{A}}}} \right)} \right|^2}$, thereby causing substantial performance degradation. To address this issue, we propose a novel active sensing-assisted communication framework that exploits the interplay between the sensing channels ${{\bf{H}}_{\rm{s}}}$ and ${\mu _{\rm{A}}}$ to enhance ground-to-UAV transmission, as detailed in the next section.

\section{Framework of Active Sensing-assisted Communication}
To fully harness the communication rate improvement enabled by wireless sensing, we propose a two-stage active sensing-assisted communication framework to maximize the achievable rate in the case of UAV jittering. As illustrated in Fig. \ref{framework}, each time slot is divided into two stages, namely stage I and stage II, with durations ${T_1}$ and ${T_2}$, respectively, satisfying ${T_1} + {T_2} = T$. In Stage I, sensing is performed to obtain high-quality AoA information based on the signals received by the active sensors. In Stage II, the transceiver is dedicated solely to communication by leveraging the AoA information acquired in Stage I. Under this framework, we propose two practical operating schemes that differ in the type of transmitted signals used in Stage I, namely communication-oriented scheme and sensing-oriented scheme, as detailed below.

\subsection{Communication-oriented Scheme}
\begin{figure}[!t]
\setlength{\abovecaptionskip}{-5pt}
\setlength{\belowcaptionskip}{-5pt}
\centering
\includegraphics[width= 0.48\textwidth]{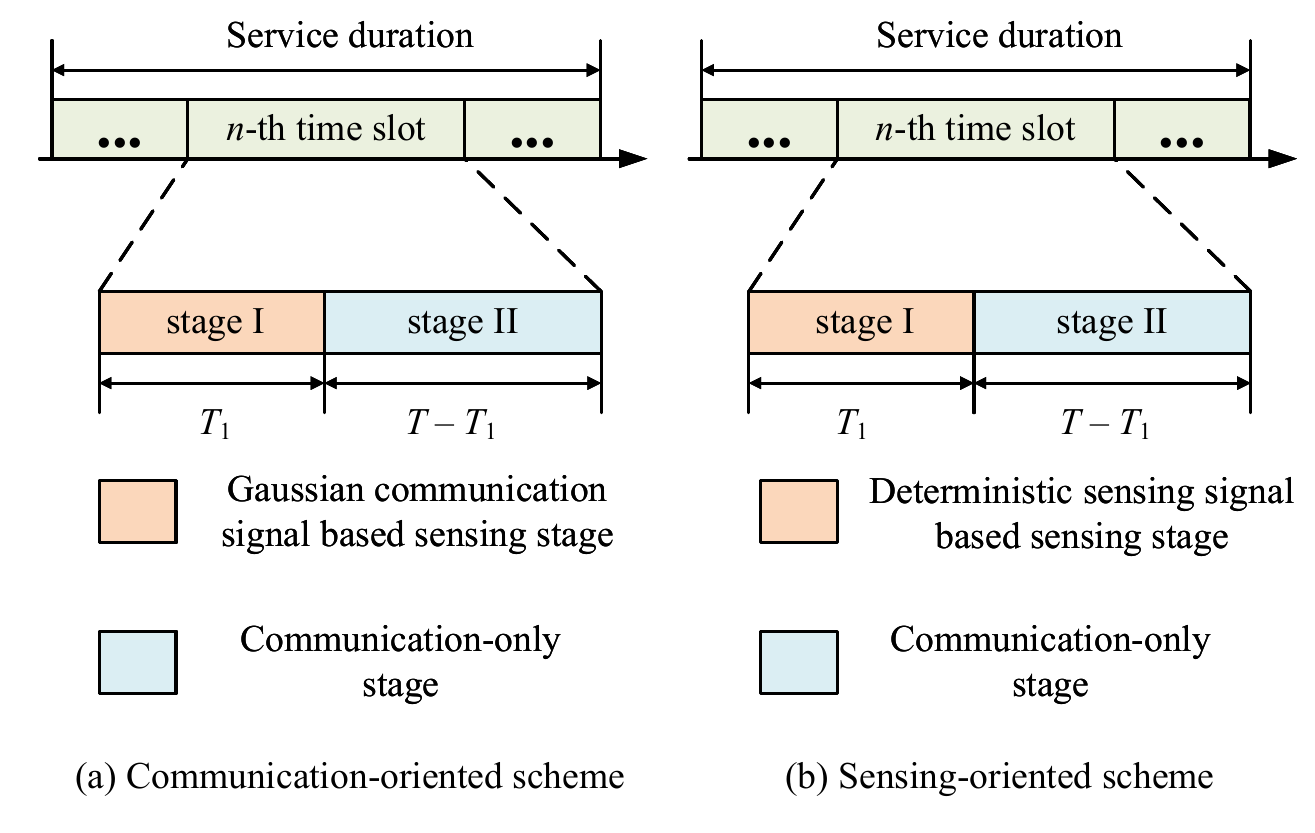}
\DeclareGraphicsExtensions.
\caption{An active sensing-assisted communication framework.}
\label{framework}
\vspace{-8pt}
\end{figure}
We first introduce the proposed communication-oriented scheme. Let ${{\cal T}_1} = \left\{ {1, \ldots ,{T_1}} \right\}$ and ${{\cal T}_2} = \left\{ {{T_1} + 1, \ldots ,T} \right\}$ denote the set of sample indices in stage I and stage II, respectively. In this scheme, the BS transmits pure Gaussian communication signals to the UAV in stage I. To maximize the received power at the UAV, the optimal transmit beamforming vector at the BS is set to ${{\bf{w}}^*}$ in \eqref{opt_tx_rx_bf}. Accordingly, the transmit signals in stage I can be written as
\begin{align}\label{Gausssian_signal_samples}
{\bf{x}}_{\rm{c}}^1\left( t \right) = \sqrt {{P_t}} {{\bf{w}}^*}{s_{\rm{c}}}\left( t \right),t \in {{\cal T}_1},
\end{align}
where ${s_{\rm{c}}}\left( t \right) \sim {\cal CN}\left( {0,1} \right)$. Hence, we have ${\bf{x}}_{\rm{c}}^1\left( t \right) \sim {\cal CN}\left( {{{\bf{0}}_{{M_t} \times 1}},{\bf{Q}}_{\rm{c}}^1} \right)$ with covariance matrix ${\bf{Q}}_{\rm{c}}^1 = {\rm{E}}\left[ {{\bf{x}}_{\rm{c}}^1\left( t \right){{\left( {{\bf{x}}_{\rm{c}}^1\left( t \right)} \right)}^H}} \right] = {P_t}{{\bf{w}}^*}{\left( {{{\bf{w}}^*}} \right)^H}$. The received signal at the active sensors of UAV is given by
\begin{align}\label{received_sensing_signal1}
{\bf{y}}_{\rm{c}}^1\left( t \right) &= {{\bf{H}}_s}{\bf{x}}_{\rm{c}}^1\left( t \right) + {{\bf{n}}_s}\left( t \right)\nonumber\\
& = \beta _{\rm{s}}^x\sqrt {{P_t}{M_t}} {{\bf{b}}_{\rm{s}}}\left( {{\mu _{\rm{A}}}} \right){s_{\rm{c}}}\left( t \right) + {{\bf{n}}_{\rm{s}}}\left( t \right),t \in {{\cal T}_1},
\end{align}
where  ${{\bf{n}}_{\rm{s}}}$ is the AWGN at the active sensors satisfying ${{\bf{n}}_{\rm{s}}} \sim {\cal CN}\left( {{\bf{0}},{\sigma ^2}{{\bf{I}}_{{M_{\rm{s}}}}}} \right)$. For ease of exposition, we stack the transmitted symbols, the received signals, and AWGN over all ${T_1}$ samples as ${\bf{s}}_{\rm{c}}^1 = {\left[ {{s_{\rm{c}}}\left( 1 \right), \ldots ,{s_{\rm{c}}}\left( {{T_1}} \right)} \right]^T}$, ${\bf{Y}}_{\rm{c}}^1 = \left[ {{\bf{y}}_{\rm{c}}^1\left( 1 \right), \ldots ,{\bf{y}}_{\rm{c}}^1\left( {{T_1}} \right)} \right]$, and ${{\bf{N}}_{\rm{s}}} = \left[ {{{\bf{n}}_{\rm{s}}}\left( 1 \right), \ldots ,{{\bf{n}}_{\rm{s}}}\left( {{T_1}} \right)} \right]$. Then, the received signal at the active sensors of UAV over the entire sensing period ${{\cal T}_1}$ is stacked as
\begin{align}\label{received_sensing_signal_sample}
{\bf{Y}}_{\rm{c}}^1 = \beta _{\rm{s}}^x\sqrt {{P_t}{M_t}} {{\bf{b}}_{\rm{s}}}\left( {{\mu _{\rm{A}}}} \right){\left( {{\bf{s}}_{\rm{c}}^1} \right)^T} + {{\bf{N}}_{\rm{s}}},
\end{align}
By vectorizing ${\bf{Y}}_{\rm{c}}^1$ and ${{\bf{N}}_{\rm{s}}}$, i.e., ${\bf{\tilde y}}_{\rm{c}}^1 = {\mathop{\rm vec}\nolimits} \left( {{\bf{Y}}_{\rm{c}}^1} \right)$ and ${{{\bf{\tilde n}}}_{\rm{s}}} = {\mathop{\rm vec}\nolimits} \left( {{{\bf{N}}_{\rm{s}}}} \right)$, we can rewrite \eqref{received_sensing_signal_sample} as
\begin{align}\label{received_sensing_signal_sample2}
{\bf{\tilde y}}_{\rm{c}}^1 &= {\mathop{\rm vec}\nolimits} \left( {{\bf{Y}}_{\rm{c}}^1} \right)\nonumber\\
& = \beta _{\rm{s}}^x\sqrt {{P_t}{M_t}} \left( {{{\bf{I}}_{T_1}} \otimes {{\bf{b}}_{\rm{s}}}\left( {{\mu _{\rm{A}}}} \right)} \right){\bf{s}}_{\rm{c}}^1 + {{{\bf{\tilde n}}}_{\rm{s}}}\nonumber\\
& = {\bf{\tilde u}}_{\rm{c}}^1 + {{{\bf{\tilde n}}}_{\rm{s}}},
\end{align}
where ${\bf{\tilde u}}_{\rm{c}}^1 = \beta _{\rm{s}}^x\sqrt {{P_t}{M_t}} \left( {{{\bf{I}}_{T_1}} \otimes {{\bf{b}}_{\rm{s}}}\left( {{\mu _{\rm{A}}}} \right)} \right){\bf{s}}_{\rm{c}}^1$. Notice from \eqref{received_sensing_signal_sample2} that the received samples ${\bf{\tilde y}}_{\rm{c}}^1$ contain useful information about ${{\mu _{\rm{A}}}}$. In the following, we design a practical estimator based on the principle of MLE.

For notational convenience, we denote ${{\bf{b}}_{\rm{s}}}\left( {{\mu _{\rm{A}}}} \right)$ as ${{\bf{b}}_{\rm{s}}}$ in the following. Accordingly, the covariance of ${{\bf{\tilde y}}_{\rm{c}}^1}$ can be expressed as
\begin{align}\label{covariance_yc1}
{\bf{R}}_c^1 &= {\mathop{\rm E}\nolimits} \left( {{\bf{\tilde y}}_{\rm{c}}^1{{\left( {{\bf{\tilde y}}_{\rm{c}}^1} \right)}^H}} \right)\nonumber\\
& = {{\bf{I}}_{T_1}} \otimes \left( {{{\left| {\beta _{\rm{s}}^x} \right|}^2}{P_t}{M_t}{{\bf{b}}_{\rm{s}}}{\bf{b}}_{\rm{s}}^H + {\sigma ^2}{{\bf{I}}_{{M_{\rm{s}}}}}} \right).
\end{align}
Based on \eqref{received_sensing_signal_sample2}, the likelihood function of the received signal is given by
\begin{align}\label{likelihood1}
{p_{\rm{c}}}\left( {{\bf{\tilde y}}_{\rm{c}}^1} \right) = \frac{{\exp \left( { - {{\left( {{\bf{\tilde y}}_{\rm{c}}^1} \right)}^H}{{\left( {{\bf{R}}_c^1} \right)}^{ - 1}}{\bf{\tilde y}}_{\rm{c}}^1} \right)}}{{{\pi ^{{M_{\rm{s}}}{T_1}}}\det \left( {{\bf{R}}_c^1} \right)}}.
\end{align}
Then, the log-likelihood function $\ln {p_{\rm{c}}}\left( {{\bf{\tilde y}}_{\rm{c}}^1} \right)$ can be expressed as
\begin{align}\label{log_likelihood1}
\ln {p_{\rm{c}}}\left( {{\bf{\tilde y}}_{\rm{c}}^1} \right) =&   - {T_1}{M_{\rm{s}}}\ln \pi {\sigma ^2}\left( {1 + {\gamma _{\rm{s}}}} \right) - \frac{1}{{{\sigma ^2}}}{\sum\nolimits_{t = 1}^{{T_1}} {\left\| {{\bf{y}}_{\rm{c}}^1\left( t \right)} \right\|} ^2}\nonumber\\
& + \frac{{{\gamma _{\rm{s}}}}}{{{\sigma ^2}{{\left\| {{{\bf{b}}_{\rm{s}}}} \right\|}^2}\left( {1 + {\gamma _{\rm{s}}}} \right)}}{\sum\nolimits_{t = 1}^{{T_1}} {\left| {{\bf{b}}_{\rm{s}}^H{\bf{y}}_{\rm{c}}^1\left( t \right)} \right|} ^2},
\end{align}
where ${\gamma _{\rm{s}}} = {P_t}{M_t}{M_{\rm{s}}}{\left| {\beta _{\rm{s}}^x} \right|^2}/{\sigma ^2}$ denotes the sensing SNR. By ignoring the constant terms in \eqref{log_likelihood1}, maximizing ${p_{\rm{c}}}\left( {{\bf{\tilde y}}_{\rm{c}}^1} \right)$ is equivalent to maximizing ${\sum\nolimits_{t = 1}^{{T_1}} {\left| {{\bf{b}}_{\rm{s}}^H{\bf{y}}_{\rm{c}}^1\left( t \right)} \right|} ^2}$. Consequently, the MLE of ${\mu _{\rm{A}}}$ is given by
\begin{align}\label{MLE_estimator}
\hat \mu _{\rm{A}}^{\rm{c}} = \mathop {\arg \max }\limits_{{\mu _{\rm{A}}} \in \left[ { - 1,1} \right]} {\sum\nolimits_{t = 1}^{{T_1}} {\left| {{\bf{b}}_{\rm{s}}^H{\bf{y}}_{\rm{c}}^1\left( t \right)} \right|} ^2},
\end{align}
which can be efficiently obtained via a one-dimensional search.

Regarding the communication task, the  Gaussian symbols $\left\{ {{s_{\rm{c}}}\left( t \right)} \right\}$ are transmitted via the BS over the entire time slot, i.e., for all $t \in {\cal T}$. Since the AoA ${{\mu _{\rm{A}}}}$ is unavailable in stage I, the receive beamforming in stage I is set based on the previous-slot AoA ${{\bar \mu }_{\rm{A}}}$, i.e., ${\bf{v}}_1^{\rm{c}} = {{\bf{b}}_{\rm{c}}}\left( {{{\bar \mu }_{\rm{A}}}} \right)/\left\| {{{\bf{b}}_{\rm{c}}}\left( {{{\bar \mu }_{\rm{A}}}} \right)} \right\|$. Consequently, the achievable rate in stage I is given by
\begin{align}\label{rate1_c}
R_1^c \!\!=\!\! {{\mathop{\rm E}\nolimits} _{{{\bar \mu }_{\rm{A}}}}}\left[ {{{\log }_2}\left( {1 \!\!+\!\! \frac{{{P_t}{M_t}{{\left| {\beta _{\rm{c}}} \right|}^2}{{\left| {{{\left( {{\bf{v}}_1^{\rm{c}}} \right)}^H}{{\bf{b}}_{\rm{c}}}\left( {{\mu _{\rm{A}}}} \right)} \right|}^2}}}{{{\sigma ^2}}}} \right)} \right],
\end{align}
where the expectation is taken over the randomness of ${{{\bar \mu }_{\rm{A}}}}$. In stage II, the estimated value of AoA $\hat \mu _{\rm{A}}^{\rm{c}}$ is available and thus the receive beamforming is set as ${\bf{v}}_2^{\rm{c}} = {{\bf{b}}_{\rm{c}}}\left( {\hat \mu _{\rm{A}}^{\rm{c}}} \right)/\left\| {{{\bf{b}}_{\rm{c}}}\left( {\hat \mu _{\rm{A}}^{\rm{c}}} \right)} \right\|$. Accordingly, the achievable rate in stage II can be expressed as
\begin{align}\label{rate2_c}
R_2^c \!\!=\!\! {{\mathop{\rm E}\nolimits} _{\hat \mu _{\rm{A}}^{\rm{c}}}}\left[ {{{\log }_2}\left( {1 \!\!+\!\! \frac{{{P_t}{M_t}{{\left| {\beta _{\rm{c}}} \right|}^2}{{\left| {{{\left( {{\bf{v}}_2^{\rm{c}}} \right)}^H}{{\bf{b}}_{\rm{c}}}\left( {{\mu _{\rm{A}}}} \right)} \right|}^2}}}{{{\sigma ^2}}}} \right)} \right].
\end{align}
Then, the overall achievable rate of the communication-oriented scheme is given by
\begin{align}\label{rate_c}
{R_{\rm{c}}} = \frac{1}{T}\left( {{T_1}R_1^c + \left( {{T} - {T_1}} \right)R_2^c} \right).
\end{align}
\subsection{Sensing-oriented Scheme}
In the sensing-oriented scheme, the BS transmits pure deterministic sensing signals in stage I, which is given by ${\bf{x}}_{\rm{s}}^1\left( t \right) = \sqrt {{P_t}} {{\bf{w}}^*}{s_{\rm{d}}}\left( t \right),t \in {{\cal T}_1}$. Different from the Gaussian symbols $\left\{ {{s_{\rm{c}}}\left( t \right)} \right\}$, the deterministic symbols $\left\{ {{s_{\rm{d}}}\left( t \right)} \right\}$ satisfying $\left| {{s_{\rm{d}}}\left( t \right)} \right| = 1$ are known at the UAV side. In stage 1, the received signals at the active sensors of UAV are given by
\begin{align}\label{received_sensing_signal2}
{\bf{y}}_{\rm{s}}^1\left( t \right) = \beta _{\rm{s}}^x\sqrt {{P_t}{M_t}} {{\bf{b}}_{\rm{s}}}{s_{\rm{d}}}\left( t \right) + {{\bf{n}}_{\rm{s}}}\left( t \right),t \in {{\cal T}_1}.
\end{align}
Similarly, the collected samples are stacked as ${\bf{\tilde y}}_{\rm{s}}^1 = {\left[ {{{\left( {{\bf{y}}_{\rm{s}}^1\left( 1 \right)} \right)}^T}, \ldots ,{{\left( {{\bf{y}}_{\rm{s}}^1\left( {{T_1}} \right)} \right)}^T}} \right]^T}$. We rewrite all the data samples in stage 1 as
\begin{align}\label{received_sensing_signal_sample3}
{\bf{\tilde y}}_{\rm{s}}^1 = \beta _{\rm{s}}^x\sqrt {{P_t}{M_t}} \left( {{{\bf{I}}_T} \otimes {{\bf{b}}_{\rm{s}}}} \right){\bf{s}}_{\rm{d}}^1 + {{{\bf{\tilde n}}}_{\rm{s}}} = {\bf{\tilde u}}_{\rm{s}}^1 + {{{\bf{\tilde n}}}_{\rm{s}}},
\end{align}
where ${\bf{s}}_{\rm{d}}^1 ={\left[ {{{s}}_{\rm{d}}\left( 1 \right), \ldots ,{{s}}_{\rm{d}}\left( {{T_1}} \right)} \right]^T}$.

Then, the MLE is carried out based on \eqref{received_sensing_signal_sample3}. Since ${\bf{s}}_{\rm{d}}^1$ is known at the UAV, the likelihood function is given by
\begin{align}\label{likelihood2}
{p_{\rm{c}}}\left( {{\bf{\tilde y}}_{\rm{s}}^1} \right) = \frac{{\exp \left( { - {{\left\| {{\bf{\tilde y}}_{\rm{s}}^1 - {\bf{u}}_s^1} \right\|}^2}/{\sigma ^2}} \right)}}{{{{\left( {\pi {\sigma ^2}} \right)}^{{M_{\rm{s}}}{T_1}}}}}.
\end{align}
Notice that the amplitude of $\beta _{\rm{s}}^x$ is known, but its phase $\arg \left( {\beta _{\rm{s}}^x} \right)$ is unknown. By ignoring constant terms in \eqref{likelihood2}, maximizing the likelihood function is equivalent to
\begin{align}\label{likelihood_problem}
\mathop {\min }\limits_{{\mu _{\rm{A}}} \in \left[ { - 1,1} \right],\arg \left( {\beta _{\rm{s}}^x} \right) \in \left[ {0,2\pi } \right]} {\left\| {{\bf{\tilde y}}_{\rm{s}}^1 - {\bf{u}}_s^1} \right\|^2}.
\end{align}
For any given ${{\mu _{\rm{A}}}}$, the optimization problem with respect to ${\arg \left( {\beta _{\rm{s}}^x} \right)}$ reduces to
\begin{align}\label{problem_phase}
\mathop {\max }\limits_{\arg \left( {\beta _{\rm{s}}^x} \right) \in \left[ {0,2\pi } \right]} {\mathop{\rm Re}\nolimits} \left\{ {{e^{ - j\arg \left( {\beta _{\rm{s}}^x} \right)}}\sum\limits_{t = 1}^{{T_1}} {s_{\rm{d}}^{*}\left( t \right){\bf{b}}_{\rm{s}}^H{\bf{y}}_{\rm{s}}^1\left( t \right)} } \right\},
\end{align}
which yields the optimal phase
\begin{align}\label{opt_phase}
\arg \left( {\beta _{\rm{s}}^x} \right) = \arg \left( {\sum\limits_{t = 1}^{{T_1}} {s_{\rm{d}}^{*}\left( t \right){\bf{b}}_{\rm{s}}^H{\bf{y}}_{\rm{s}}^1\left( t \right)} } \right).
\end{align}
By substituting \eqref{opt_phase} into \eqref{likelihood_problem}, the MLE of ${{\mu _{\rm{A}}}}$ is
\begin{align}\label{MLE_estimator1}
\hat \mu _{\rm{A}}^{\rm{s}} = \mathop {\arg \max }\limits_{{\mu _{\rm{A}}} \in \left[ { - 1,1} \right]} \left| {\sum\limits_{t = 1}^{{T_1}} {s_{\rm{d}}^{*}\left( t \right){\bf{b}}_{\rm{s}}^H{\bf{y}}_{\rm{s}}^1\left( t \right)} } \right|.
\end{align}

Since the deterministic sensing signal is transmitted by the BS in stage I, no information bits are conveyed. In stage II,  the BS transmits Gaussian communication symbols  $\left\{ {{s_{\rm{c}}}\left( t \right)} \right\},t \in {{\cal T}_2}$. By using the obtained $\hat \mu _{\rm{A}}^{\rm{s}}$, the receive beamforming in stage II is set as ${\bf{v}}_2^{\rm{s}} = {{\bf{b}}_{\rm{c}}}\left( {\hat \mu _{\rm{A}}^{\rm{s}}} \right)/\left\| {{{\bf{b}}_{\rm{c}}}\left( {\hat \mu _{\rm{A}}^{\rm{s}}} \right)} \right\|$. The corresponding average achievable rate is
\begin{align}\label{rate2_s}
R_2^s \!\!=\!\! {{\mathop{\rm E}\nolimits} _{\hat \mu _{\rm{A}}^{\rm{s}}}}\left[ {{{\log }_2}\left( {1 \!\!+\!\! \frac{{{P_t}{M_t}{{\left| {\beta _{\rm{c}}} \right|}^2}{{\left| {{{\left( {{\bf{v}}_2^{\rm{s}}} \right)}^H}{{\bf{b}}_{\rm{c}}}\left( {{\mu _{\rm{A}}}} \right)} \right|}^2}}}{{{\sigma ^2}}}} \right)} \right].
\end{align}
Consequently, the overall achievable rate of the sensing-oriented scheme is given by
\begin{align}\label{rate_s}
{R_{\rm{s}}} = \frac{1}{T}\left( {T - {T_1}} \right)R_2^{\rm{s}}.
\end{align}
\begin{rem}
The key difference between the communication-oriented scheme and sensing-oriented scheme lies in the type of transmitted signals used for sensing in stage I. In particular, the MLE of the communication-oriented scheme in \eqref{MLE_estimator} only relies on the received samples $\left\{ {{\bf{y}}_{\rm{s}}^1\left( t \right)} \right\}$ since the transmitted symbols $\left\{ {s_{\rm{c}}\left( t \right)} \right\}$ are not available at the UAV. In contrast, the sensing-oriented scheme exploits both the transmitted symbols $\left\{ {s_{\rm{d}}\left( t \right)} \right\}$ and received samples. Consequently, the sensing-oriented scheme is expected to achieve higher sensing accuracy at the cost of providing no information rate in stage I. The following section aims to provide a theoretical performance comparison of the two schemes.
\end{rem}

\section{Performance Analysis}
To gain a deep understanding of the gap between the proposed schemes and the system performance limit, we analytically characterize the performance of the proposed two schemes under the active sensing-assisted communication framework.
\subsection{Characterization of Achievable Rates}
Notice that the communication rates of both schemes are highly dependent on the sensing quality in stage I. Therefore, we first characterize the sensing quality of AoA for each scheme. Generally, mean squared error (MSE) is a typical metric for sensing accuracy. However, closed-form analytical expressions for ${\rm{E}}\left\{ {{{\left| {\hat \mu _{\rm{A}}^{\rm{c}} - {\mu _{\rm{A}}}} \right|}^2}} \right\}$ and ${\rm{E}}\left\{ {{{\left| {\hat \mu _{\rm{A}}^{\rm{s}} - {\mu _{\rm{A}}}} \right|}^2}} \right\}$ are intractable to obtain. Nevertheless, we derive closed-form CRBs in the following lemma, which serve as lower bounds for the two unbiased estimators.
\begin{lem}
The CRBs of ${\hat \mu _{\rm{A}}^{\rm{c}}}$ and ${\hat \mu _{\rm{A}}^{\rm{s}}}$ are given by
\begin{align}\label{CRB_AoA}
&{\rm{CRB}}\left( {\hat \mu _{\rm{A}}^{\rm{c}}} \right) = \frac{{{\sigma ^2}}}{{2{P_t}{T_1}{{\left| {\beta _{\rm{s}}^x} \right|}^2}{M_t}{{\left\| {{{{\bf{\dot b}}}_{\rm{s}}}} \right\|}^2}}}\left( {1 + \frac{1}{{{\gamma _{\rm{s}}}}}} \right),\nonumber\\
&{\rm{CRB}}\left( {\hat \mu _{\rm{A}}^{\rm{s}}} \right) = \frac{{{\sigma ^2}}}{{2{P_t}{T_1}{{\left| {\beta _{\rm{s}}^x} \right|}^2}{M_t}{{\left\| {{{{\bf{\dot b}}}_{\rm{s}}}} \right\|}^2}}},
\end{align}
where ${\gamma _{\rm{s}}} = \frac{{{P_t}{{\left| {\beta _{\rm{s}}^x} \right|}^2}{M_t}{M_{\rm{s}}}}}{{{\sigma ^2}}}$ represents the sensing SNR and ${{{\bf{\dot b}}}_{\rm{s}}}$ denotes the partial derivative of ${{\bf{b}}_s}$ with respect to ${\mu _{\rm{A}}}$.
\begin{proof}
Please refer to Appendix A.
\end{proof}
\end{lem}

Lemma 1 shows  that ${\rm{CRB}}\left( {\hat \mu _{\rm{A}}^{\rm{c}}} \right) > {\rm{CRB}}\left( {\hat \mu _{\rm{A}}^{\rm{s}}} \right)$, which implies that the sensing-oriented scheme achieves higher sensing accuracy than the communication-oriented scheme. This result agrees with the intuition, as the sensing-oriented scheme fully exploits the known transmitted symbols. More specifically, ${\rm{CRB}}\left( {\hat \mu _{\rm{A}}^{\rm{c}}} \right)$ suffers an extra factor of $\left( {1 + 1/{\gamma _{\rm{s}}}} \right)$ compared to ${\rm{CRB}}\left( {\hat \mu _{\rm{A}}^{\rm{s}}} \right)$.

Inspired by Lemma 1, the distributions of the estimated AoA values, i.e, ${\hat \mu _{\rm{A}}^{\rm{c}}}$ and ${\hat \mu _{\rm{A}}^{\rm{s}}}$, are naturally approximated as $\hat \mu _{\rm{A}}^{\rm{c}} \sim {\cal N}\left( {{\mu _{\rm{A}}},\sigma _{\hat \mu _{\rm{A}}^{\rm{c}}}^2} \right)$ and $\hat \mu _{\rm{A}}^{\rm{s}} \sim {\cal N}\left( {{\mu _{\rm{A}}},\sigma _{\hat \mu _{\rm{A}}^{\rm{s}}}^2} \right)$, respectively, where $\sigma _{\hat \mu _{\rm{A}}^{\rm{c}}}^2 = {\rm{CRB}}\left( {\hat \mu _{\rm{A}}^{\rm{c}}} \right)$ and $\sigma _{\hat \mu _{\rm{A}}^{\rm{s}}}^2 = {\rm{CRB}}\left( {\hat \mu _{\rm{A}}^{\rm{s}}} \right)$. Using these distributions, we obtain upper bounds on the achievable rates ${R_{\rm{c}}}$ and ${R_{\rm{s}}}$, denoted by ${{\bar R}_{\rm{c}}}$ and ${{\bar R}_{\rm{s}}}$, respectively. The results are provided in the following theorem.

\begin{thm}
For any given duration ${{T_1}}$ of stage I, ${{\bar R}_{\rm{c}}}$ and ${{\bar R}_{\rm{s}}}$ are given by
\begin{align}\label{rate_upperbound}
&{\bar R_{\rm{c}}} \!=\! \frac{1}{T}\left[ {{T_1}{{\log }_2}\left( {1 \!+\! \Gamma \left( {\sigma _\mu ^2} \right)} \right) \!+\! \left( {T \!-\! {T_1}} \right){{\log }_2}\left( {1 \!+\! \Gamma \left( {\sigma _{\hat \mu _{\rm{A}}^{\rm{c}}}^2} \right)} \right)} \right],\nonumber\\
&{{\bar R}_{\rm{s}}} = \frac{1}{T}\left[ {\left( {T - {T_1}} \right){{\log }_2}\left( {1 + \Gamma \left( {\sigma _{\hat \mu _{\rm{A}}^{\rm{s}}}^2} \right)} \right)} \right],
\end{align}
where
\begin{align}\label{SNR_expectation}
\Gamma \left( x \right) \!=\! \frac{{{P_t}{M_t}}}{{{\sigma ^2}{{\left| {{\beta _c}} \right|}^{ - 2}}}}\left( {1 \!+\! \frac{2}{{{M_r}}}\sum\limits_{m = 1}^{{M_r} - 1} {\left( {{M_r} \!-\! m} \right){e^{ - \frac{{{m^2}{\pi ^2}x}}{2}}}} } \right).
\end{align}
\end{thm}
\begin{proof}
Please refer to Appendix B.
\end{proof}

Theorem 1 establishes a functional relationship between the overall achievable rate and system parameters (e.g., ${T_1}$, ${{M_t}}$, ${{M_r}}$, and ${M_{\rm{s}}}$). In Theorem 1, ${{\bar R}_{{\rm{c}}1}} = {\log _2}\left( {1 + \Gamma \left( {\sigma _\mu ^2} \right)} \right)$ represents the achievable rate in stage I of the communication-oriented scheme, which is independent of ${T_1}$ and decreases with ${\sigma _\mu ^2}$. Furthermore, ${\log _2}\left( {1 + \Gamma \left( {\sigma _\xi ^2} \right)} \right),\xi  \in \left\{ {\hat \mu _{\rm{A}}^{\rm{c}},\hat \mu _{\rm{A}}^{\rm{s}}} \right\}$ represents the achievable rate in stage II for each scheme, which decreases as ${\sigma _\xi ^2}$ increases.
\begin{rem}
Note that $\sigma _{\hat \mu _{\rm{A}}^{\rm{c}}}^2$ and $\sigma _{\hat \mu _{\rm{A}}^{\rm{s}}}^2$ scale inversely linearly with respect to ${T_1}$, which implies that the sensing accuracy of the AoA improves with more time allocated to stage I, thereby increasing the achievable rate in stage II. However,  a larger ${T_1}$ leaves less time for stage II, which is used solely for communication, thus unfavourably reducing the total achievable rate. Consequently, a fundamental tradeoff exists between the sensing accuracy in stage I and the communication time resource in stage II for maximizing the overall achievable rate.
\end{rem}

\subsection{Optimization of Time Allocation}
Remark 2 reveals a fundamental tradeoff between sensing and communication by adjusting ${T_1}$ to enhance the overall achievable rate. To capture the maximum achievable rate of both the sensing-oriented scheme and communication-oriented scheme, we formulate the optimization problems as
\begin{align}\label{opt_time1}
\mathop {\max }\limits_{{T_1}}~ {{\bar R}_{\rm{s}}} ~~~{\rm{s.t.}}~0 \le {T_1} \le T,
\end{align}
\begin{align}\label{opt_time2}
\mathop {\max }\limits_{{T_1}}~ {{\bar R}_{\rm{c}}} ~~~{\rm{s.t.}}~0 \le {T_1} \le T.
\end{align}
The optimal solutions of both problems \eqref{opt_time1} and \eqref{opt_time2} can be obtained efficiently via a one-dimensional search over $\left[ {0,T} \right]$. However, such numerical-based solutions may not provide concrete insights into system design. To draw useful insights, we obtain tractable approximations for ${{\bar R}_{\rm{c}}}$ and ${{\bar R}_{\rm{s}}}$ in the following proposition.
\begin{pos}
The approximations for ${{\bar R}_{\rm{c}}}$ and ${{\bar R}_{\rm{s}}}$ are
\begin{align}\label{approximation_rate}
{{\tilde R}_c} \!=\! \frac{1}{T}\left[ {{T_1}{{\bar R}_{{\rm{c}}1}} \!+\! \left( {T \!-\! {T_1}} \right){{\tilde R}_{{\rm{c}}2}}} \right],{{\tilde R}_{\rm{s}}} \!=\! \frac{1}{T}\left( {T \!-\! {T_1}} \right){{\tilde R}_{{\rm{s}}2}},
\end{align}
where
\begin{align}\label{approximation_rate_temp}
{{\tilde R}_i} = {\log _2}\left( {1 + \frac{{{P_t}{{\left| {\beta _{\rm{c}}} \right|}^2}{M_t}{M_r}}}{{{\sigma ^2}\sqrt {1 + {A_i}/{T_1}} }}} \right),i \in \left\{ {{\rm{c}}2,{\rm{s}}2} \right\}
\end{align}
with ${A_{{\rm{c2}}}} = \frac{{{\pi ^2}\left( {M_r^2 - 1} \right){\sigma ^2}\left( {1 + \frac{1}{{{\gamma _s}}}} \right)}}{{12{P_t}{{\left| {\beta _{\rm{s}}^x} \right|}^2}{M_t}{{\left\| {{{{\bf{\dot b}}}_{\rm{s}}}} \right\|}^2}}},{A_{{\rm{s2}}}} = \frac{{{\pi ^2}\left( {M_r^2 - 1} \right){\sigma ^2}}}{{12{P_t}{{\left| {\beta _{\rm{s}}^x} \right|}^2}{M_t}{{\left\| {{{{\bf{\dot b}}}_{\rm{s}}}} \right\|}^2}}}$.
\end{pos}
\begin{proof}
Please refer to Appendix C.
\end{proof}

Let ${\gamma _{{\rm{c}}}} = \frac{{{P_t}{{\left| {\beta _{\rm{c}}} \right|}^2}{M_t}{M_r}}}{{{\sigma ^2}}}$ denote the maximum communication SNR at the receiver in the absence of AoA estimation error, and let $u\left( {{A_i},{T_1}} \right) = \frac{1}{{\sqrt {1 + {A_i}/{T_1}} }}$, $i \in \left\{ {{\rm{c}}2,{\rm{s}}2} \right\}$ denote the corresponding beamforming gain loss due to AoA estimation errors. Then, the approximate rates in Proposition 1 can be rewritten as
\begin{align}\label{approximation_rate1}
&{{\tilde R}_{\rm{c}}} = \frac{1}{T}\left[ {{T_1}{{\bar R}_{{\rm{c}}1}} + \left( {T - {T_1}} \right){{\log }_2}\left( {1 + {\gamma _{\rm{c}}}u\left( {{A_{{\rm{c2}}}},{T_1}} \right)} \right)} \right],\nonumber\\
&{{\tilde R}_{\rm{s}}} = \frac{1}{T}\left( {T - {T_1}} \right){\log _2}\left( {1 + {\gamma _{\rm{c}}}u\left( {{A_{{\rm{s2}}}},{T_1}} \right)} \right).
\end{align}
Note that \eqref{approximation_rate1} exhibits a tractable functional relationship between the achievable rate and the time allocated to stage I, i.e, ${{T_1}}$. By replacing  $\left\{ {{{\bar R}_{\rm{c}}},{{\bar R}_{\rm{s}}}} \right\}$ with $\left\{ {{{\tilde R}_{\rm{c}}},{{\tilde R}_{\rm{s}}}} \right\}$ in problems \eqref{opt_time1} and \eqref{opt_time2}, the corresponding time allocation problems can be reformulated as
\begin{align}\label{opt_time_app1}
\mathop {\max }\limits_{{T_1}}~ {{\tilde R}_{\rm{s}}} ~~~{\rm{s.t.}}~0 \le {T_1} \le T,
\end{align}
\vspace{-5pt}
\begin{align}\label{opt_time_app2}
\mathop {\max }\limits_{{T_1}}~ {{\tilde R}_{\rm{c}}} ~~~{\rm{s.t.}}~0 \le {T_1} \le T.
\end{align}
Denote $T_{1{\rm{s}}}^* = \mathop {\arg \max }\limits_{{T_1} \in \left[ {0,T} \right]} {{\tilde R}_{\rm{s}}}$ and $T_{1{\rm{c}}}^* = \mathop {\arg \max }\limits_{{T_1} \in \left[ {0,T} \right]} {{\tilde R}_{\rm{c}}}$ as the optimal solutions to problems \eqref{opt_time_app1} and \eqref{opt_time_app2}, respectively. Then, we derive the semi-closed-form expressions for $T_{1{\rm{s}}}^*$ and $T_{1{\rm{c}}}^*$ in the following proposition.
\begin{pos}
For problem \eqref{opt_time_app1}, $T_{1{\rm{s}}}^*$ is the unique root of the following equation
\begin{align}\label{equation_s}
{F_{\rm{s}}}\left( t \right) &\buildrel \Delta \over =  - \ln \left( {1 + {\gamma _{\rm{c}}}u\left( {{A_{{\rm{s2}}}},t} \right)} \right) + \frac{{\left( {T - t} \right){\gamma _{\rm{c}}}{A_{{\rm{s2}}}}{u^3}\left( {{A_{{\rm{s2}}}},t} \right)}}{{2{t^2}\left( {1 + {\gamma _{\rm{c}}}u\left( {{A_{{\rm{s2}}}},t} \right)} \right)}}\nonumber\\
& = 0,
\end{align}
located in $\left[ {0,T} \right]$. For problem \eqref{opt_time_app2}, $T_{1{\rm{c}}}^*$ is the unique root of the equation
\begin{align}\label{equation_c}
{F_{\rm{c}}}\left( t \right){\rm{ }}\buildrel \Delta \over = {{\bar R}_{{\rm{c}}1}}\ln 2  &- \ln \left( {1 \!+\! {\gamma _{\rm{c}}}u\left( {{A_{{\rm{c2}}}},t} \right)} \right)\nonumber\\
&+ \frac{{\left( {T \!-\! t} \right){\gamma _{\rm{c}}}{A_{{\rm{c2}}}}{u^3}\left( {{A_{{\rm{c2}}}},t} \right)}}{{2{t^2}\left( {1 + {\gamma _{\rm{c}}}u\left( {{A_{{\rm{c2}}}},{t}} \right)} \right)}} = 0
\end{align}
if and only if
\begin{align}\label{sensing_condition}
{{\bar R}_{{\rm{c}}1}} < {\log _2}\left( {1 + {\gamma _{\rm{c}}}u\left( {{A_{{\rm{c2}}}},T} \right)} \right).
\end{align}
Otherwise, $T_{1{\rm{c}}}^* = T$.
\end{pos}
\begin{proof}
Please refer to Appendix D.
\end{proof}

For the communication-oriented scheme, condition \eqref{sensing_condition} in Proposition 2 explicitly answers  when active sensing is needed to maximize the overall achievable rate. Recall that ${\bar R_{{\rm{c}}1}}$ is the achievable rate in stage I, in which the receive beamforming is set based on the AoA obtained in the previous time slot. The term ${\log _2}\left( {1 + {\gamma _{\rm{c}}}u\left( {{A_{{\rm{c2}}}},T} \right)} \right)$ in \eqref{sensing_condition} represents the achievable rate under AoA estimation error when the entire time  $T$ is devoted to sensing. The condition ${{\bar R}_{{\rm{c}}1}} \ge {\log _2}\left( {1 + {\gamma _{\rm{c}}}u\left( {{A_{{\rm{c2}}}},T} \right)} \right)$ indicates that the initial rate in stage I is already higher than the maximum rate achievable by using the AoA obtained through active sensing. In this case, the AoA sensed in stage I does not need to be exploited in stage II, thereby leading to $T_{1{\rm{c}}}^* = T$. Moreover, based on \eqref{equation_s} and \eqref{equation_c}, $T_{1{\rm{s}}}^*$ and $T_{1{\rm{c}}}^*$ can be obtained efficiently by a bisection search over $\left[ {0,T} \right]$ since both ${F_{\rm{s}}}\left( t \right)$ and ${F_{\rm{c}}}\left( t \right)$ are monotonically decreasing functions.

\subsection{Sensing-oriented Versus Communication-oriented Scheme}
Proposition 2 lays the foundation for a theoretical performance comparison between the communication-oriented scheme and sensing-oriented scheme. Let $\tilde R_{\rm{s}}^*$ and $\tilde R_{\rm{c}}^*$ denote the optimal objective values of problems \eqref{opt_time_app1} and \eqref{opt_time_app2}, respectively. A sufficient condition under which the communication-oriented scheme outperforms the sensing-oriented scheme is given by the following theorem.
\begin{thm}
At the optimal solutions of problems \eqref{opt_time_app1} and \eqref{opt_time_app2}, $\tilde R_{\rm{c}}^* > \tilde R_{\rm{s}}^*$ always holds provided that
\begin{align}\label{sufficient_condition}
{{\bar R}_{{\rm{c}}1}} > {C_0}\left( {\frac{{{A_{{\rm{c2}}}}}}{{{A_{{\rm{s2}}}}}} - 1} \right),
\end{align}
where ${C_0} = {\log _2}\left( {1 + \frac{{{P_t}{M_t}{M_r}{{\left| {{\beta _{\rm{c}}}} \right|}^2}}}{{{\sigma ^2}}}} \right)$.
\end{thm}
\begin{proof}
Please refer to Appendix E.
\end{proof}

\begin{rem}
Condition \eqref{sufficient_condition} in Theorem 2 exhibits an interesting threshold-based structure. In particular, ${C_0}$ represents the maximum achievable rate in the absence of AoA estimation errors. Define ${\eta _{{\rm{c2}}}} = 1/{A_{{\rm{c}}2}}$ and ${\eta _{{\rm{s2}}}} = 1/{A_{{\rm{s}}2}}$ as the sensing efficiency of the communication-oriented scheme and the sensing-oriented scheme, respectively. In terms of sensing efficiency, the relative gain of the sensing-oriented scheme over the communication-oriented scheme is $\frac{{{\eta _{{\rm{s2}}}} - {\eta _{{\rm{c2}}}}}}{{{\eta _{{\rm{c2}}}}}} = \frac{{{A_{{\rm{c}}2}}}}{{{A_{{\rm{s}}2}}}} - 1 = \frac{1}{{{\gamma _{\rm{s}}}}}$, which decreases as the sensing SNR ${\gamma _{\rm{s}}}$ increases. Theorem 2 unveils that the communication-oriented scheme always outperforms the sensing-oriented scheme if the initial rate in stage I, ${{\bar R}_{{\rm{c}}1}}$,  exceeds a threshold that depends on sensing efficiency, namely ${C_0}\left( {\frac{{{A_{{\rm{c2}}}}}}{{{A_{{\rm{s2}}}}}} - 1} \right) = \frac{{{C_0}}}{{{\gamma _{\rm{s}}}}}$. Intuitively, condition \eqref{sufficient_condition} is more likely to hold when either the initial rate in stage I is higher or the sensing SNR is higher.
\end{rem}

\subsection{Analysis in Asymptotic Regime}
To gain more useful insights, we further study the performance of the proposed active sensing-assisted communication framework several asymptotic regimes.

\subsubsection{High Time Slot Duration $T$ Regime}
We first consider the case where the time slot duration $T$ is  asymptotically large, which corresponds to the practical scenario of a low UAV attitude variation frequency. To gain a deep understanding on the behavior of $T_{1{\rm{s}}}^*$ and $T_{1{\rm{c}}}^*$ in this regime, we obtain the following proposition.
\begin{pos}
In the asymptotically large $T$ regime, i.e., $T \to \infty$, the asymptotically optimal solutions to problems \eqref{opt_time_app1} and \eqref{opt_time_app1} are given by
\begin{align}\label{asymptotically_optimal_solution}
&T_{1{\rm{s}}}^* = \sqrt {\frac{{{\gamma _{\rm{c}}}{A_{{\rm{c}}2}}}}{{2\ln 2\left( {1 + {\gamma _{\rm{c}}}} \right){C_0}}}} \sqrt T ,\nonumber\\
&T_{1{\rm{c}}}^* = \sqrt {\frac{{{\gamma _{\rm{c}}}{A_{{\rm{s}}2}}}}{{2\ln 2\left( {1 + {\gamma _{\rm{c}}}} \right)\left( {{C_0} - {{\bar R}_{{\rm{c}}1}}} \right)}}} \sqrt T .
\end{align}
\end{pos}
\begin{proof}
Please refer to Appendix F.
\end{proof}

For both the sensing-oriented and communication-oriented schemes, we observe that the asymptotically optimal time allocated to stage I scales on the order of $T_{1\xi }^* \sim {\cal O}\left( {\sqrt T } \right),\xi  \in \left\{ {{\rm{s,c}}} \right\}$. Notably, $T_{1\xi }^*/T$ characterizes the sensing overhead. This result is promising for both schemes, as it demonstrates that the sensing overhead becomes negligible under the proposed active sensing-assisted communication framework.

\subsubsection{High Transmit Power Regime}
We next consider the case of asymptotically large transmit power. In this regime, we analytically quantify the relative performance gap between the proposed schemes and the system fundamental performance limit. The result is given in the following proposition.
\begin{pos}
In the asymptotically large transmit power regime, i.e., ${P_t} \to \infty$, we have
\begin{align}\label{performance_loss_limit}
\mathop {\lim }\limits_{{P_t} \to \infty } \frac{{{C_0} - \tilde R_{\rm{s}}^*}}{{\tilde R_{\rm{s}}^*}} = 0,\mathop {\lim }\limits_{{P_t} \to \infty } \frac{{{C_0} - \tilde R_{\rm{c}}^*}}{{\tilde R_{\rm{c}}^*}} = 0.
\end{align}
\end{pos}
\begin{proof}
Please refer to Appendix G.
\end{proof}

Recall that ${C_0}$ denotes the maximum achievable rate in the absence of AoA estimation errors, which can be viewed as the rate limit of the considered system. Therefore, Proposition 4 unveils a favorable property of our proposed active sensing-assisted communication framework, which demonstrates that the achievable rates of both schemes can approach this fundamental rate limit in the high transmit power regime.


\section{Numerical Results}
In this section, we provide numerical results to characterize the performance of the proposed schemes and to gain insights into the implementation of the active sensing-assisted communication framework. The BS and the UAV are located at $\left( {0,0,0} \right)$ meter (m) and $\left( {300,0,100} \right)$ m, respectively. The path-loss exponent of the BS-UAV link is set to 2. The signal attenuation  at a reference distance of 1 m is $ - 50$ dB, and the noise power is ${\sigma ^2} =  - 80$ dBm. Unless otherwise stated, the remaining  system parameters are set as follows: $T = 500$ symbols, ${M_t} = 10$, ${M_r} = 32$, ${M_{\rm{s}}} = 4$, $\sigma _\mu ^2 = 0.01$, and ${P_{\max }} = 10$ dBm, respectively.

\begin{figure}[t!]
\centering
\includegraphics[width=3.0in]{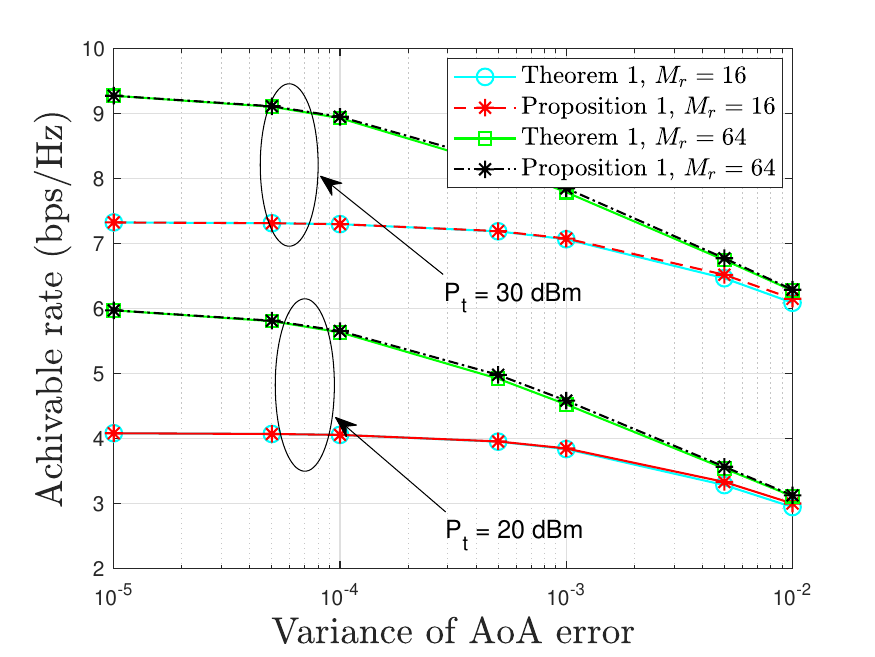}
\caption{{Achievable rate versus the variance of AoA.}}
\label{rate_accuracy}
\vspace{-8pt}
\end{figure}

\subsection{Accuracy Evaluation of Rate Analysis}
Before evaluating the performance of the proposed designs, we first verify the effectiveness of the approximate achievable rate derived in Proposition 1. To this end, we plot both the approximations in Proposition 1 and the exact results in Theorem 1 versus the variance of the AoA error under different system parameters in Fig. \ref{rate_accuracy}. First, it is observed from Fig. \ref{rate_accuracy} that the approximations match the exact results quite well, which confirms the validity of the derived approximations. Besides, it is expected that the achievable rates decrease as the AoA error variance increases due to the reduction in receive beamforming gains. Compared to the case with ${M_r} = 16$, the degradation in achievable rate with respect to the AoA error variance is more pronounced for large ${M_r}$, highlighting the importance of  high sensing accuracy for improving communication quality when a large antenna array is employed.

\subsection{Achievable Rate Comparison}
To validate the effectiveness of the proposed designs, we evaluate and compare the performance of the following schemes: 1) \textbf{Upper bound}: The AoA error is omitted, and thus the achievable rate is given by ${C_0} = {\log _2}\left( {1 + \frac{{{P_t}{M_t}{M_r}{{\left| {{\beta _{\rm{c}}}} \right|}^2}}}{{{\sigma ^2}}}} \right)$, which serves as a performance upper bound; 2) \textbf{Sensing-oriented scheme}: Deterministic sensing signals are transmitted in stage I under the proposed active sensing-assisted communication framework; 3) \textbf{Communication-oriented scheme}: Gaussian communication signals are employed in stage I under the same framework; 4) \textbf{Fixed time allocation}: The time allocated to stage I is set to ${T_1} = 0.2T$ under the sensing-oriented scheme; 5) \textbf{Outdated receive beamforming}: The receive beamforming is set based on the outdated AoA obtained from the previous time slot.

\begin{figure}[t!]
\centering
\includegraphics[width=3.0in]{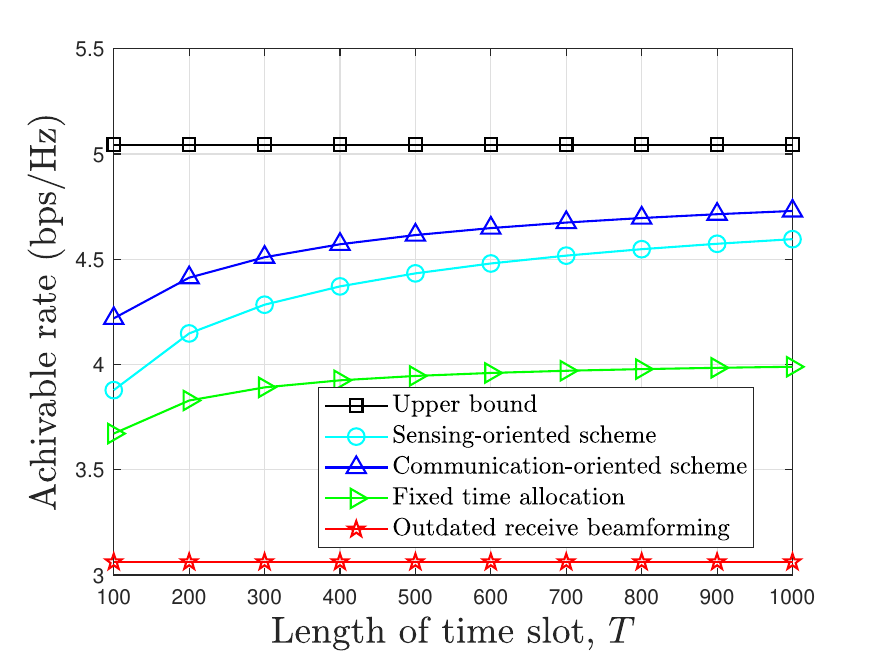}
\caption{{Achievable rate versus the length of time slot.}}
\label{rate_time}
\vspace{-8pt}
\end{figure}

\subsubsection{Effect of Time Slot Length}
In Fig. \ref{rate_time}, we plot the system achievable rate versus the length of the time slot. It is observed that our proposed communication-oriented and sensing oriented schemes significantly improve the achievable rate compared to the fixed-time-allocation and outdated-receive-beamforming baselines. In particular, the performance gains over benchmark schemes become more pronounced as $T$ increases. This result is expected since the proposed schemes allow a more flexible tradeoff to balance the sensing accuracy in stage I and communication quality in stage II. Besides, we observe that the communication-oriented scheme always outperforms the sensing-oriented scheme under the considered setup. In this setup, the sensing SNR is ${\gamma _{\rm{s}}} = \frac{{{P_t}{{\left| {\beta _{\rm{s}}^x} \right|}^2}{M_t}{M_{\rm{s}}}}}{{{\sigma ^2}}} = 4$. By comparing the rate of outdated receive beamforming and upper bound in Fig. \ref{rate_time}, the condition ${{\bar R}_{{\rm{c}}1}} > {C_0}/{\gamma _{\rm{s}}}$ in Theorem 2 is always satisfied. This result demonstrates that the rate improvement in stage II achieved by the higher sensing accuracy of the sensing-oriented scheme may not compensate for its rate loss in stage I relative to the communication-oriented scheme.

\begin{figure}[t!]
\centering
\includegraphics[width=3.0in]{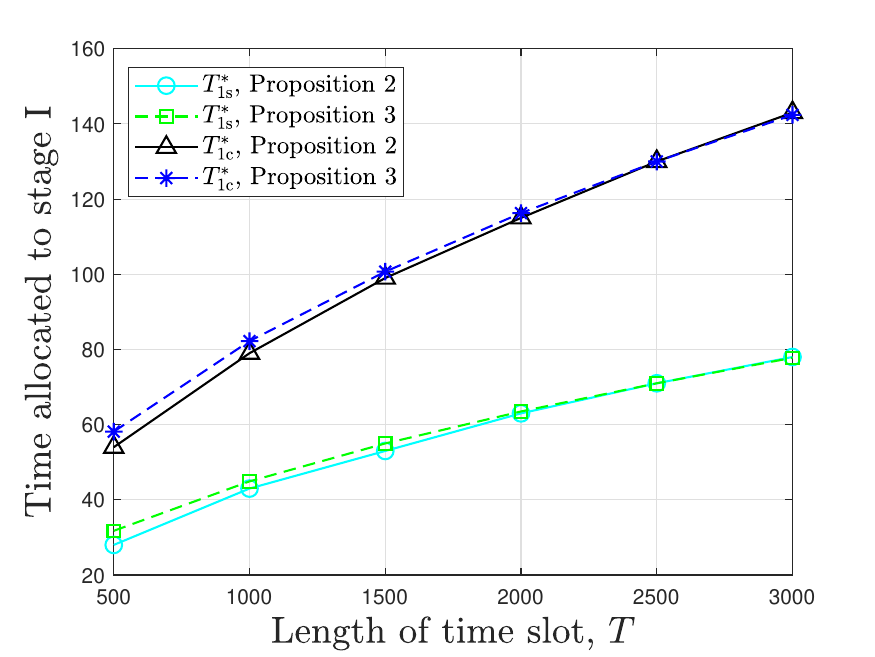}
\caption{{Illustration of the optimal time allocation.}}
\label{sensing_time}
\vspace{-8pt}
\end{figure}
Then, we study the behavior of the optimal time allocated to stage I by plotting $\left\{ {T_{1{\rm{s}}}^*,T_{1{\rm{c}}}^*} \right\}$ derived in Proposition 2 and Proposition 3 versus the length of the time slot. It is observed that the gap between the numerically obtained optimal values and the analytical approximations is tight over the entire range of $T$, especially when $T$ becomes large. The result verifies the effectiveness of the scaling law $T_{1\xi }^* \sim {\cal O}\left( {\sqrt T } \right),\xi  \in \left\{ {{\rm{s,c}}} \right\}$ analyzed in Section IV, which implies that the proposed schemes offer the advantage of low sensing overhead. In addition, one can observe that the optimal time allocated to stage I for the communication-oriented scheme is larger than that of the sensing-oriented scheme. This result is expected since the communication-oriented scheme requires more time in stage I to compensate for the inherent loss in sensing accuracy compared to the sensing-oriented scheme.

\begin{figure}[t!]
\centering
\includegraphics[width=3.0in]{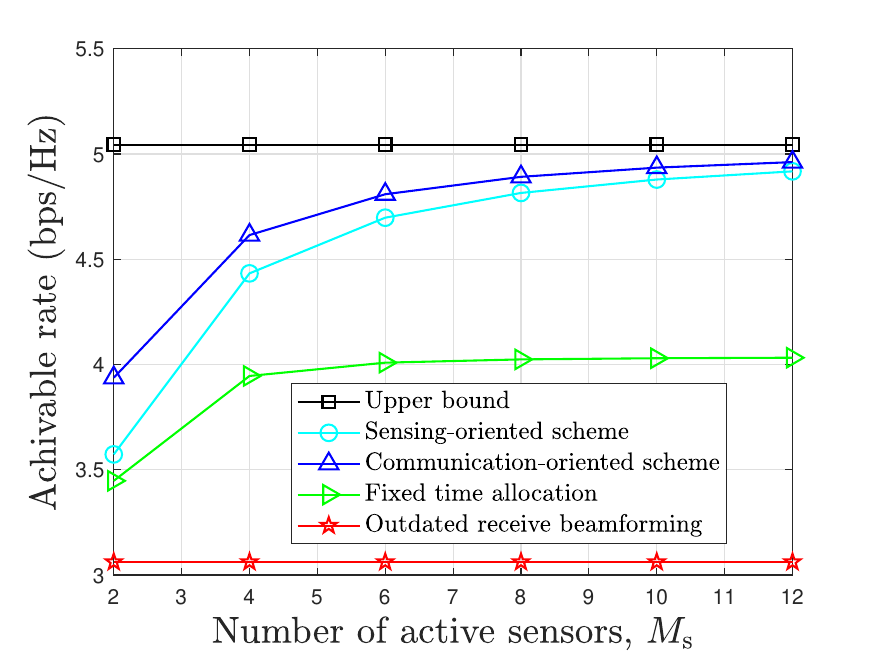}
\caption{{Achievable rate versus the number of active sensors.}}
\label{rate_sensors}
\vspace{-8pt}
\end{figure}

\subsubsection{Effect of the Number of Active Sensors}
In Fig. \ref{rate_sensors}, we study the effect of the number of active sensors by plotting the achievable rate of all schemes versus ${M_{\rm{s}}}$. It is observed that the achievable rates of the schemes that utilize stage I for sensing increase with ${M_{\rm{s}}}$. The reason is that more active sensors provide both array gain and a sampling rate, enabling high sensing accuracy in stage I, which in turn significantly improve the receive beamforming gain in stage II. Note that the CRB of AoA estimation decays with ${M_{\rm{s}}}$ on the order of ${\cal O}\left( {1/M_{\rm{s}}^3} \right)$ since ${\left\| {{{{\bf{\dot b}}}_s}} \right\|^2}$ grows with ${M_{\rm{s}}}$ as ${\cal O}\left( {M_{\rm{s}}^3} \right)$. With a large ${M_{\rm{s}}}$, high-quality AoA estimation can be obtained with a short sensing duration, allowing the rates of the proposed schemes to approach the upper bound  in the large ${M_{\rm{s}}}$ regime. Consequently, the performance gap between the proposed schemes and benchmark schemes become more pronounced as ${M_{\rm{s}}}$ increases.

\begin{figure}[t!]
\centering
\includegraphics[width=3.0in]{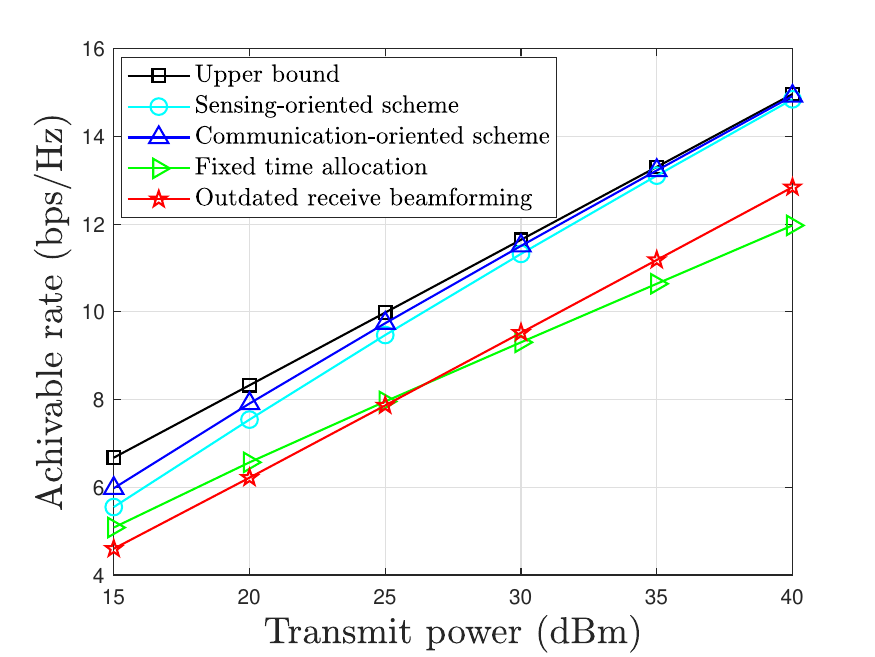}
\caption{{Achievable rate versus transmit power.}}
\label{rate_power}
\vspace{-8pt}
\end{figure}

\subsubsection{Effect of Transmit Power}
In Fig. \ref{rate_power}, we study the effect of the transmit power by plotting the achievable rates of all schemes versus ${P_t}$. First, the proposed schemes significantly outperform the benchmark schemes, and the performance gains become more evident as ${P_t}$ increases. Moreover, it is noticed that the rates of both the communication-oriented scheme and sensing-oriented scheme approach the performance upper bound with increasing ${P_t}$. In particular, the performance gap between the proposed schemes and the jittering-free upper bound becomes negligible when the transmit power is greater than 35 dBm, which agrees with our analysis in Proposition 4. In the large transmit power regime, both the sensing time devoted to stage I and the resulting AoA estimation error become small, thereby approaching the fundamental performance limit. In addition, the rate of the scheme with fixed time allocation performs even worse than the scheme with outdated receive beamforming at high ${P_t}$, highlighting the importance of proper time allocation to balance the sensing accuracy in stage I and the time utilization for information transmission in stage II.

\subsection{Discussion on Antenna Deployment}
\begin{figure}[t!]
\centering
\includegraphics[width=3.0in]{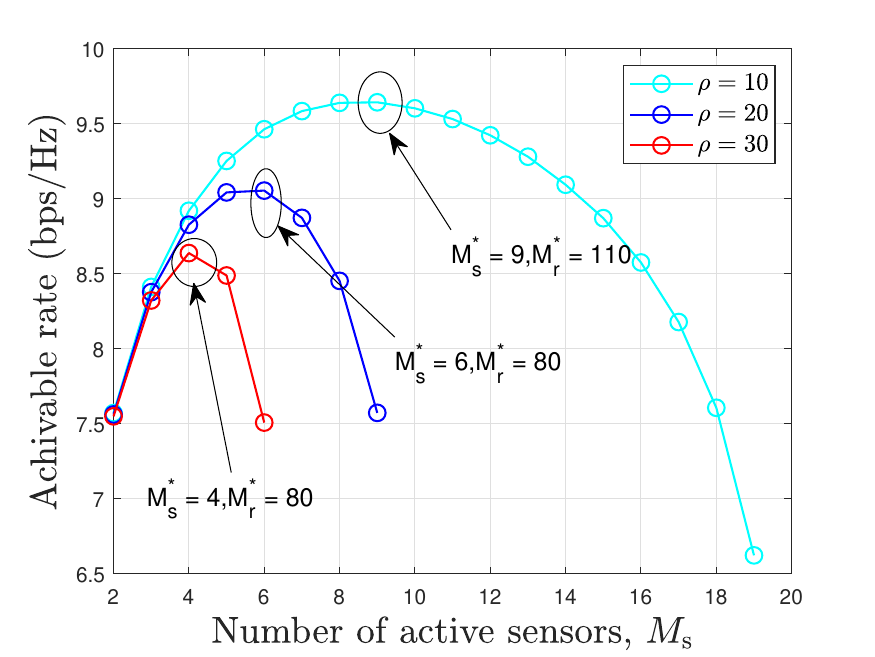}
\caption{{Illustration of the impact of antenna deployment.}}
\label{deployment}
\vspace{-8pt}
\end{figure}
Since both the active sensors and the receive antenna elements are deployed on the UAV, an interesting problem is to determine the optimal deployment combination $\left\{ {M_{\rm{s}}^*,M_r^*} \right\}$ that maximizes the system rate under a given total deployment cost. Let ${c_{\rm{E}}}$ denote the unit cost of a receive antenna element, and let ${c_{\rm{s}}} = \rho {c_{\rm{E}}}$ denote the cost of an active sensor, where $\rho  \gg 1$ is the cost ratio between an active sensor and an antenna element. For a given $\left\{ {{M_{\rm{s}}},{M_r}} \right\}$, the total deployment cost is ${C_{{\rm{tot}}}} = \left( {\rho {M_{\rm{s}}} + {M_r}} \right){c_{\rm{E}}}$.

Under a fixed total deployment cost ${C_{{\rm{tot}}}} = 200{c_{\rm{E}}}$, Fig. \ref{deployment} illustrates the achievable rate of the sensing-oriented scheme for different values of $\rho $. It is observed that the achievable rate first increases and then decreases with ${{M_{\rm{s}}}}$. Although increasing ${{M_{\rm{s}}}}$ improves sensing accuracy, over-allocating spatial resources to active sensors leaves fewer resources for deploying receive antenna elements, thereby reducing the peak array gain. This result highlights the necessity of configuring $\left\{ {M_{\rm{s}}^*,M_r^*} \right\}$ to balance the spatial sensing and communication resources, which is essential for further improving the system achievable rate. Intuitively, the system tends to deploy more active sensors when  $\rho$ is small.

\section{Conclusion}
This paper proposed a novel two-stage active sensing-assisted communication framework tailored for ground-to-UAV links with jittering. Under this framework, two schemes were designed to harness sensing for improving communication performance, namely the communication-oriented scheme and the sensing-oriented scheme. For both schemes, we derived closed-form expressions for their achievable rates, unveiling a fundamental tradeoff between sensing and communication quality across the two stages by tuning the time allocated to the first stage. The optimal time allocation that maximizes the overall rate was obtained in semi-closed-form expressions. Based on these results, we unveiled a sufficient condition under which the communication-oriented scheme outperforms the sensing-oriented scheme, which exhibits an interesting threshold-based structure. Asymptotic analysis demonstrated that the performance loss of the proposed schemes relative to the jitter-free upper bound approaches zero in the high transmit power regime, thereby highlighting their asymptotic optimality.

\section*{Appendix A: \textsc{Proof of Lemma 1}}
First, we focus on deriving the CRB for ${\hat \mu _{\rm{A}}^{\rm{c}}}$. Under the sensing-oriented scheme, the likelihood function of the received signal can be rewritten as
\begin{align}\label{likelihood1_rewritten}
{p_{\rm{c}}}\left( {{\bf{\tilde y}}_{\rm{c}}^1} \right) = \prod\limits_{t = 1}^{{T_1}} {\frac{1}{{{\pi ^{{M_{\rm{s}}}}}\det \left( {{\bf{\Sigma }}_c^1} \right)}}} {e^{ - {{\left( {{\bf{y}}_{\rm{c}}^1\left( t \right)} \right)}^H}{{\left( {{\bf{\Sigma }}_c^1} \right)}^{ - 1}}{\bf{y}}_{\rm{c}}^1\left( t \right)}},
\end{align}
where ${\bf{\Sigma }}_c^1 = {\left| {\beta _{\rm{s}}^x} \right|^2}{P_t}{M_t}{{\bf{b}}_{\rm{s}}}{\bf{b}}_{\rm{s}}^H + {\sigma ^2}{{\bf{I}}_{{M_{\rm{s}}}}}$. Based on \eqref{likelihood1_rewritten}, the Fisher information of ${\mu _{\rm{A}}}$ can be derived as
\begin{align}\label{Fisher_information_c}
{F_{\hat \mu _{\rm{A}}^{\rm{c}}}}&  \!\!=\!\!  - {\mathop{\rm E}\nolimits} \left\{ {\frac{{{\partial ^2}\ln {p_{\rm{c}}}\left( {{\bf{\tilde y}}_{\rm{c}}^1} \right)}}{{{\partial ^2}{\mu _{\rm{A}}}}}} \right\}\nonumber\\
& = - \sum\limits_{t = 1}^{{T_1}} {{\mathop{\rm E}\nolimits} \left\{ {\frac{\partial }{{\partial {\mu _{\rm{A}}}}}\left[ {{{\left( {{\bf{y}}_{\rm{c}}^1\left( t \right)} \right)}^H}{{\left( {{\bf{\Sigma }}_c^1} \right)}^{ - 1}}\frac{{\partial {\bf{\Sigma }}_c^1}}{{\partial {\mu _{\rm{A}}}}}{{\left( {{\bf{\Sigma }}_c^1} \right)}^{ - 1}}{\bf{y}}_{\rm{c}}^1\left( t \right)} \right]} \right\}}\nonumber\\
& ~~~+ {T_1}\frac{\partial }{{\partial {\mu _{\rm{A}}}}}\left[ {{{\left( {{\bf{\Sigma }}_c^1} \right)}^{ - 1}}{\mathop{\rm Tr}\nolimits} \left( {\frac{{\partial {\bf{\Sigma }}_c^1}}{{\partial {\mu _{\rm{A}}}}}} \right)} \right].
\end{align}
Let ${\bf{M}} = {\left( {{\bf{\Sigma }}_c^1} \right)^{ - 1}}\frac{{\partial {\bf{\Sigma }}_c^1}}{{\partial {\mu _{\rm{A}}}}}{\left( {{\bf{\Sigma }}_c^1} \right)^{ - 1}}$ and then the first term in \eqref{Fisher_information_c} can be calculated as
\begin{align}\label{first_term}
&- \sum\nolimits_{t = 1}^{{T_1}} {{\rm{E}}\left\{ {\frac{\partial }{{\partial {\mu _{\rm{A}}}}}\left[ {{{\left( {{\bf{y}}_{\rm{c}}^1\left( t \right)} \right)}^H}{\bf{My}}_{\rm{c}}^1\left( t \right)} \right]} \right\}}\nonumber\\
& = 2{T_1}{\mathop{\rm Tr}\nolimits} \left( {{{\left( {{\bf{\Sigma }}_c^1} \right)}^{ - 1}}\frac{{\partial {\bf{\Sigma }}_c^1}}{{\partial {\mu _{\rm{A}}}}}{{\left( {{\bf{\Sigma }}_c^1} \right)}^{ - 1}}\frac{{\partial {\bf{\Sigma }}_c^1}}{{\partial {\mu _{\rm{A}}}}}} \right)\nonumber\\
& ~~~- {T_1}{\mathop{\rm Tr}\nolimits} \left( {{{\left( {{\bf{\Sigma }}_c^1} \right)}^{ - 1}}\frac{{{\partial ^2}{\bf{\Sigma }}_c^1}}{{{\partial ^2}\mu _{\rm{A}}^2}}} \right).
\end{align}
The second term in \eqref{Fisher_information_c} is derived as
\begin{align}\label{second_term}
&{T_1}\frac{\partial }{{\partial {\mu _{\rm{A}}}}}\left[ {{{\left( {{\bf{\Sigma }}_c^1} \right)}^{ - 1}}{\rm{Tr}}\left( {\frac{{\partial {\bf{\Sigma }}_c^1}}{{\partial {\mu _{\rm{A}}}}}} \right)} \right]\nonumber\\
& = {T_1}{\mathop{\rm Tr}\nolimits} \left( { - {{\left( {{\bf{\Sigma }}_c^1} \right)}^{ - 1}}\left( {\frac{{\partial {\bf{\Sigma }}_c^1}}{{\partial {\mu _{\rm{A}}}}}{{\left( {{\bf{\Sigma }}_c^1} \right)}^{ - 1}}\frac{{\partial {\bf{\Sigma }}_c^1}}{{\partial {\mu _{\rm{A}}}}} + \frac{{{\partial ^2}{\bf{\Sigma }}_c^1}}{{{\partial ^2}\mu _{\rm{A}}^2}}} \right)} \right).
\end{align}
By plugging \eqref{first_term} and \eqref{second_term} into \eqref{Fisher_information_c}, we have
\begin{align}\label{Fisher_information_c_sim}
{F_{\hat \mu _{\rm{A}}^{\rm{c}}}} = {T_1}{\mathop{\rm Tr}\nolimits} \left( {{{\left( {{\bf{\Sigma }}_c^1} \right)}^{ - 1}}\frac{{\partial {\bf{\Sigma }}_c^1}}{{\partial {\mu _{\rm{A}}}}}{{\left( {{\bf{\Sigma }}_c^1} \right)}^{ - 1}}\frac{{\partial {\bf{\Sigma }}_c^1}}{{\partial {\mu _{\rm{A}}}}}} \right),
\end{align}
In \eqref{Fisher_information_c_sim}, ${{{\left( {{\bf{\Sigma }}_c^1} \right)}^{ - 1}}}$ and $\frac{{\partial {\bf{\Sigma }}_c^1}}{{\partial {\mu _{\rm{A}}}}}$ are given by
\begin{align}\label{terms}
&{\left( {{\bf{\Sigma }}_c^1} \right)^{ - 1}} = \frac{1}{{{\sigma ^2}}}\left( {{\bf{I}} - \frac{{{{\left| {\beta _{\rm{s}}^x} \right|}^2}{P_t}{M_t}/{\sigma ^2}}}{{1 + {{\left| {\beta _{\rm{s}}^x} \right|}^2}{P_t}{M_t}{M_{\rm{s}}}/{\sigma ^2}}}{{\bf{b}}_{\rm{s}}}{\bf{b}}_{\rm{s}}^H} \right)\nonumber\\
&\frac{{\partial {\bf{\Sigma }}_c^1}}{{\partial {\mu _{\rm{A}}}}} = {\left| {\beta _{\rm{s}}^x} \right|^2}{P_t}{M_t}\left( {{{{\bf{\dot b}}}_{\rm{s}}}{\bf{b}}_{\rm{s}}^H + {{\bf{b}}_{\rm{s}}}{\bf{\dot b}}_{\rm{s}}^H} \right).
\end{align}
By substituting \eqref{terms} into \eqref{Fisher_information_c_sim} and using the result ${\bf{\dot b}}_{\rm{s}}^H{{\bf{b}}_{\rm{s}}} = 0$, we obtain
\begin{align}\label{Fisher_information_c_result}
{F_{\hat \mu _{\rm{A}}^{\rm{c}}}} = \frac{{2{T_1}{{\left| {\beta _{\rm{s}}^x} \right|}^4}{{\left( {{P_t}{M_t}} \right)}^2}{M_{\rm{s}}}{{\left\| {{{{\bf{\dot b}}}_{\rm{s}}}} \right\|}^2}}}{{{\sigma ^4} + {\sigma ^2}{{\left| {\beta _{\rm{s}}^x} \right|}^2}{P_t}{M_t}{M_{\rm{s}}}}}.
\end{align}
Based on \eqref{Fisher_information_c_result}, the CRB of ${\hat \mu _{\rm{A}}^{\rm{c}}}$ is ${\mathop{\rm CRB}\nolimits} \left( {\hat \mu _{\rm{A}}^{\rm{c}}} \right) = 1/{F_{\hat \mu _{\rm{A}}^{\rm{c}}}}$, which leads to the result in Lemma 1.

The CRB of ${\hat \mu _{\rm{A}}^{\rm{s}}}$ can be derived by following the similar steps, which are omitted here due to its brevity.

\section*{Appendix B: \textsc{Proof of Theorem 1}}
First, we aim to derive the upper bound of ${R_{\rm{s}}}$. According to Jensen's inequality, it holds that ${\mathop{\rm E}\nolimits} \left\{ {{{\log }_2}\left( {1 + X} \right)} \right\} \le {\log _2}\left( {1 + {\mathop{\rm E}\nolimits} \left\{ X \right\}} \right)$. Consequently, $R_2^{\rm{s}}$ is upper bounded by
\begin{align}\label{rate2s_upper}
\bar R_2^{\rm{s}} = {\log _2}\left( {1 + \frac{{{P_t}{M_t}{{\rm{E}}_{\hat \mu _{\rm{A}}^{\rm{s}}}}\left\{ {{{\left| {{{\left( {{\bf{v}}_2^{\rm{s}}} \right)}^H}{{\bf{b}}_{\rm{c}}}\left( {{\mu _{\rm{A}}}} \right)} \right|}^2}} \right\}}}{{{\sigma ^2}{{\left| {\beta _{\rm{c}}} \right|}^{ - 2}}}}} \right).
\end{align}
In the following, we focus on the derivation of ${{{\rm{E}}_{\hat \mu _{\rm{A}}^{\rm{s}}}}\left\{ {{{\left| {{{\left( {{\bf{v}}_2^{\rm{s}}} \right)}^H}{{\bf{b}}_{\rm{c}}}\left( {{\mu _{\rm{A}}}} \right)} \right|}^2}} \right\}}$. By plugging ${\bf{v}}_2^{\rm{s}} = {{\bf{b}}_{\rm{c}}}\left( {\hat \mu _{\rm{A}}^{\rm{s}}} \right)/\left\| {{{\bf{b}}_{\rm{c}}}\left( {\hat \mu _{\rm{A}}^{\rm{s}}} \right)} \right\|$ into ${{{\left( {{\bf{v}}_2^{\rm{s}}} \right)}^H}{{\bf{b}}_{\rm{c}}}\left( {{\mu _{\rm{A}}}} \right)}$, we have
\begin{align}\label{bf_gain}
{\left| {{{\left( {{\bf{v}}_2^{\rm{s}}} \right)}^H}{{\bf{b}}_{\rm{c}}}\left( {{\mu _{\rm{A}}}} \right)} \right|^2} = \frac{1}{{{M_r}}}{\left| {\frac{{\sin \left( {\pi {M_r}{\Delta _{\hat \mu _{\rm{A}}^{\rm{s}}}}/2} \right)}}{{\sin \left( {\pi {\Delta _{\hat \mu _{\rm{A}}^{\rm{s}}}}/2} \right)}}} \right|^2},
\end{align}
where ${\Delta _{\hat \mu _{\rm{A}}^{\rm{s}}}} = \hat \mu _{\rm{A}}^{\rm{s}} - {\mu _{\rm{A}}}$ with ${\Delta _{\hat \mu _{\rm{A}}^{\rm{s}}}} \sim {\cal N}\left( {0,\sigma _{\hat \mu _{\rm{A}}^{\rm{s}}}^2} \right)$. Then, ${{{\rm{E}}_{\hat \mu _{\rm{A}}^{\rm{s}}}}\left\{ {{{\left| {{{\left( {{\bf{v}}_2^{\rm{s}}} \right)}^H}{{\bf{b}}_{\rm{c}}}\left( {{\mu _{\rm{A}}}} \right)} \right|}^2}} \right\}}$ can be derived as
\begin{align}\label{bf_gain_expectation}
& {{\rm{E}}_{{\Delta _{\hat \mu _{\rm{A}}^{\rm{s}}}}}}\left\{ {\frac{1}{{{M_r}}}{{\left| {\frac{{\sin \left( {\pi {M_r}{\Delta _{\hat \mu _{\rm{A}}^{\rm{s}}}}/2} \right)}}{{\sin \left( {\pi {\Delta _{\hat \mu _{\rm{A}}^{\rm{s}}}}/2} \right)}}} \right|}^2}} \right\}\nonumber\\
& = \frac{1}{{{M_r}}}\sum\nolimits_{m = 1 - {M_r}}^{{M_r} - 1} {\left( {N - \left| m \right|} \right){{\rm{E}}_{{\Delta _{\hat \mu _{\rm{A}}^{\rm{s}}}}}}\left\{ {{e^{jm\pi {\Delta _{\hat \mu _{\rm{A}}^{\rm{s}}}}}}} \right\}}\nonumber\\
& \mathop  = \limits^{\left( a \right)}  \frac{2}{{{M_r}}}\sum\nolimits_{m = 1}^{{M_r} - 1} {\left( {{M_r} - m} \right){e^{ - \frac{1}{2}{m^2}{\pi ^2}\sigma _{\hat \mu _{\rm{A}}^{\rm{s}}}^2}}},
\end{align}
where (a) follows by ${\rm{E}}\left\{ {{e^{jtX}}} \right\} = {e^{ - \frac{1}{2}{t^2}{\sigma ^2}}}$ with $X \sim {\cal N}\left( {0,{\sigma ^2}} \right)$. By plugging \eqref{bf_gain_expectation} into \eqref{rate2s_upper}, $\bar R_2^{\rm{s}}$ can be obtained. Accordingly, ${{\bar R}_{\rm{s}}}$ is obtained in \eqref{rate_upperbound}. By following the similar steps, ${{\bar R}_{\rm{s}}}$ can be derived, which is omitted due to its brevity.

\section*{Appendix C: \textsc{Proof of Proposition 1}}
By taking the approximation for ${\bar R_{\rm{s}}}$ as an example, we first approximate the array gain in stage II as
\begin{align}\label{bf_gain_approximation}
{\left| {{{\left( {{\bf{v}}_2^{\rm{s}}} \right)}^H}{{\bf{b}}_{\rm{c}}}\left( {{\mu _{\rm{A}}}} \right)} \right|^2} &= \frac{1}{{{M_r}}}{\left| {\frac{{\sin \left( {\pi {M_r}{\Delta _{\hat \mu _{\rm{A}}^{\rm{s}}}}/2} \right)}}{{\sin \left( {\pi {\Delta _{\hat \mu _{\rm{A}}^{\rm{s}}}}/2} \right)}}} \right|^2}\nonumber\\
&\mathop  \approx \limits^{\left( a \right)} \frac{1}{{{M_r}}}{\left[ {\frac{{\left( {\frac{{\pi {M_r}{\Delta _{\hat \mu _{\rm{A}}^{\rm{s}}}}}}{2}} \right) - \frac{1}{6}{{\left( {\frac{{\pi {M_r}{\Delta _{\hat \mu _{\rm{A}}^{\rm{s}}}}}}{2}} \right)}^3}}}{{\left( {\frac{{\pi {\Delta _{\hat \mu _{\rm{A}}^{\rm{s}}}}}}{2}} \right) - \frac{1}{6}{{\left( {\frac{{\pi {\Delta _{\hat \mu _{\rm{A}}^{\rm{s}}}}}}{2}} \right)}^3}}}} \right]^2}\nonumber\\
&\mathop  \approx \limits^{\left( b \right)} {M_r}\left[ {1 - \frac{{{\pi ^2}\left( {M_r^2 - 1} \right){{\left( {{\Delta _{\hat \mu _{\rm{A}}^{\rm{s}}}}} \right)}^2}}}{{12}}} \right]\nonumber\\
&\mathop  \approx \limits^{\left( c \right)} {M_r}{e^{ - \frac{{{\pi ^2}\left( {M_r^2 - 1} \right){{\left( {{\Delta _{\hat \mu _{\rm{A}}^{\rm{s}}}}} \right)}^2}}}{{12}}}},
\end{align}
where (a) follows by taking the high order Taylor expansions for both the numerator and denominator, (b) follows by the result of ${\left( {1 - \varepsilon } \right)^{ - 1}} \approx 1 + \varepsilon  $ for a small $\varepsilon$ and omitting the fourth-order terms of ${{\Delta _{\hat \mu _{\rm{A}}^{\rm{s}}}}}$, (c) follows by ${e^{ - \varepsilon }} \approx 1 - \varepsilon$. Then, we approximate ${\mathop{\rm E}\nolimits} \left\{ {{{\left| {{{\left( {{\bf{v}}_2^{\rm{s}}} \right)}^H}{{\bf{b}}_{\rm{c}}}\left( {{\mu _{\rm{A}}}} \right)} \right|}^2}} \right\}$ as
\begin{align}\label{bf_gain_app_expectation}
{\mathop{\rm E}\nolimits} \left\{ {{{\left| {{{\left( {{\bf{v}}_2^{\rm{s}}} \right)}^H}{{\bf{b}}_{\rm{c}}}\left( {{\mu _{\rm{A}}}} \right)} \right|}^2}} \right\}& \approx {\mathop{\rm E}\nolimits} \left\{ {{M_r}{e^{ - \frac{{{\pi ^2}\left( {M_r^2 - 1} \right){{\left( {{\Delta _{\hat \mu _{\rm{A}}^{\rm{s}}}}} \right)}^2}}}{{12}}}}} \right\}\nonumber\\
&\mathop  = \limits^{\left( a \right)} \frac{{{M_r}}}{{\sqrt {1 + \frac{{{\pi ^2}\left( {M_r^2 - 1} \right)\sigma _{\hat \mu _{\rm{A}}^{\rm{s}}}^2}}{6}} }},
\end{align}
where (a) follows by ${\Delta _{\hat \mu _{\rm{A}}^{\rm{s}}}} \sim {\cal N}\left( {0,\sigma _{\hat \mu _{\rm{A}}^{\rm{s}}}^2} \right)$. By using the result $\sigma _{\hat \mu _{\rm{A}}^{\rm{s}}}^2 = {\mathop{\rm CRB}\nolimits} \left( {\hat \mu _{\rm{A}}^{\rm{s}}} \right)$ and plugging \eqref{bf_gain_app_expectation} into the expression of $\bar R_2^{\rm{s}}$, ${{\tilde R}_s}$ in Proposition 1 can be obtained. Notice that ${{\tilde R}_c}$ can be obtained in a similar way and thus we omit the details due to its brevity.

\section*{Appendix D: \textsc{Proof of Proposition 2}}
First, we focus on deriving $T_{1{\rm{s}}}^*$. For notational simplicity, we use $u$ to replace ${u\left( {{A_{{\rm{s2}}}},{T_1}} \right)}$. By taking the first-order derivative of ${{\tilde R}_{\rm{s}}}$ with respect to ${T_1}$, we obtain
\begin{align}\label{first_order_deriivative}
\frac{{\partial {{\tilde R}_{\rm{s}}}}}{{\partial {T_1}}} =  - {\log _2}\left( {1 + {\gamma _{\rm{c}}}u} \right) + \frac{{\left( {T - {T_1}} \right){\gamma _{\rm{c}}}{A_{{\rm{s2}}}}{u^3}}}{{2T_1^2\left( {1 + {\gamma _{\rm{c}}}u} \right)\ln 2}}.
\end{align}
Let $H\left( {{T_1}} \right) = \partial {{\tilde R}_{\rm{s}}}/\partial {T_1}$. Then, we have
\begin{align}\label{boundary_result}
&\mathop {\lim }\limits_{{T_1} \to 0} H\left( {{T_1}} \right) = \frac{{T{\gamma _{\rm{c}}}{A_{{\rm{s2}}}}{u^3}}}{{2T_1^2\left( {1 + {\gamma _{\rm{c}}}u} \right)\ln 2}} =  + \infty > 0 \nonumber\\
&H\left( T \right) =  - {\log _2}\left( {1 + {{\left( {1 + \frac{{{A_{{\rm{s2}}}}}}{T}} \right)}^{ - 1/2}}} \right) < 0.
\end{align}
Next, we aim to prove that $H\left( {{T_1}} \right)$ is a decreasing function with respect to ${{T_1}}$. Let $H\left( {{T_1}} \right) = {H_1}\left( {{T_1}} \right) + H\left( {{T_2}} \right)$ with ${H_1}\left( {{T_1}} \right) =  - {\log _2}\left( {1 + {\gamma _{\rm{c}}}u} \right)$ and $H\left( {{T_2}} \right) = \frac{{\left( {T - {T_1}} \right){\gamma _{\rm{c}}}{A_{{\rm{s2}}}}{u^3}}}{{2T_1^2\left( {1 + {\gamma _{\rm{c}}}u} \right)\ln 2}}$. It is evident that $u$ increases with the increase of ${{T_1}}$. Hence, ${H_1}\left( {{T_1}} \right)$ is a decreasing function with respect to ${{T_1}}$. By exploiting the result ${T_1} = \frac{{{A_{{\rm{s2}}}}{u^2}}}{{1 - {u^2}}}$, ${H_2}\left( {{T_1}} \right)$ is given by
\begin{align}\label{H2_temp}
&{H_2}\left( {{T_1}} \right) = G\left( u \right) = \frac{1}{{\ln 2}}\left( {T - \frac{{{A_{{\rm{s2}}}}{u^2}}}{{1 - {u^2}}}} \right)\frac{{{\gamma _{\rm{c}}}{{\left( {1 - {u^2}} \right)}^2}}}{{2{A_{{\rm{s2}}}}u\left( {1 + {\gamma _{\rm{c}}}u} \right)}}\nonumber\\
& = \frac{{{\gamma _{\rm{c}}}}}{{2{A_{{\rm{s2}}}}\ln 2}}\left( {T\left( {1 - {u^2}} \right) - {A_{{\rm{s2}}}}{u^2}} \right)\frac{{{\gamma _{\rm{c}}}\left( {1 - {u^2}} \right)}}{{u\left( {1 + {\gamma _{\rm{c}}}u} \right)}}.
\end{align}
Notice that ${H_2}\left( {{T_1}} \right)$ is monotonically decreasing is equivalent to that $G\left( u \right)$ is monotonically decreasing. It can be easily shown that both the terms $\left( {T\left( {1 - {u^2}} \right) - {A_{{\rm{s2}}}}{u^2}} \right)$ and $\frac{{\left( {1 - {u^2}} \right)}}{{u\left( {1 + {\gamma _{\rm{c}}}u} \right)}}$ decrease with $u$. Since $0 < u \le {\left( {1 + \frac{{{A_{{\rm{s2}}}}}}{T}} \right)^{ - 1/2}}$, we have $T\left( {1 - {u^2}} \right) - {A_{{\rm{s2}}}}{u^2} \ge 0$ and $\frac{{\left( {1 - {u^2}} \right)}}{{u\left( {1 + {\gamma _{\rm{c}}}u} \right)}} > 0$, which leads to that $G\left( u \right)$ is a decreasing function and $G\left( u \right) \ge 0$. Consequently, $H\left( {{T_1}} \right)$ decreases with respect to ${T_1}$. By combining the results in \eqref{boundary_result}, $\frac{{\partial {{\tilde R}_{\rm{s}}}}}{{\partial {T_1}}} = 0$ has a unique solution located in $\left[ {0,T} \right]$, which lead to $T_{1{\rm{s}}}^*$ in Proposition 2.

Then, we consider the derivation of $T_{1{\rm{c}}}^*$. By taking the first-order derivative of ${{\tilde R}_{\rm{c}}}$ with respect to ${T_1}$, we obtain
\begin{align}\label{first_order_deriivative_c}
\frac{{\partial {{\tilde R}_{\rm{c}}}}}{{\partial {T_1}}} = {{\bar R}_{{\rm{c1}}}} &- {\log _2}\left( {1 + {\gamma _{\rm{c}}}u\left( {{A_{{\rm{c2}}}},{T_1}} \right)} \right)\nonumber\\
&+ \frac{{\left( {T - {T_1}} \right){\gamma _{\rm{c}}}{A_{{\rm{c2}}}}{u^3}\left( {{A_{{\rm{c2}}}},{T_1}} \right)}}{{2T_1^2\left( {1 + {\gamma _{\rm{c}}}u\left( {{A_{{\rm{c2}}}},{T_1}} \right)} \right)\ln 2}}.
\end{align}
By following the similar steps, it can be shown that $\frac{{\partial {{\tilde R}_{\rm{c}}}}}{{\partial {T_1}}}$ decreases with respect to ${{T_1}}$ and $\frac{{\partial {{\tilde R}_{\rm{c}}}}}{{\partial {T_1}}}\left| {_{{T_1} = 0}} \right. = \infty  > 0$. Next, we have $\frac{{\partial {{\tilde R}_{\rm{c}}}}}{{\partial {T_1}}}\left| {_{{T_1} = T}} \right. = {{\bar R}_{{\rm{c1}}}} - {\log _2}\left( {1 + {\gamma _{\rm{c}}}u\left( {{A_{{\rm{c2}}}},{T_1}} \right)} \right)$. If the condition ${{\bar R}_{{\rm{c1}}}} \le {\log _2}\left( {1 + {\gamma _{\rm{c}}}u\left( {{A_{{\rm{c2}}}},{T_1}} \right)} \right)$ satisfies, $\frac{{\partial {{\tilde R}_{\rm{c}}}}}{{\partial {T_1}}} = 0$ has a unique solution located in $\left[ {0,T} \right]$, which lead to $T_{1{\rm{s}}}^*$ in \eqref{equation_c} provided in Proposition 2. Otherwise, $\frac{{\partial {{\tilde R}_{\rm{c}}}}}{{\partial {T_1}}} \ge 0$ always holds, which leads to that $T_{1{\rm{c}}}^* = T$.

\section*{Appendix E: \textsc{Proof of Theorem 2}}
Let ${h_{\rm{c}}}\left( {t,A} \right) = t{{\bar R}_{{\rm{c}}1}} + \left( {T - t} \right){\log _2}\left( {1 + {\gamma _{\rm{c}}}u\left( {A,t} \right)} \right)$ and ${h_{\rm{s}}}\left( {t,A} \right) = \left( {T - t} \right){\log _2}\left( {1 + {\gamma _{\rm{c}}}u\left( {A,t} \right)} \right)$. First, we aim to compare ${{\tilde R}_{\rm{c}}}$ and ${{\tilde R}_{\rm{s}}}$ at the optimal solution of \eqref{opt_time_app1}. Notice that ${{\tilde R}_{\rm{c}}}\left| {_{{T_1} = T_{1{\rm{s}}}^*} = {h_{\rm{c}}}\left( {T_{1{\rm{s}}}^*,{A_{{\rm{c2}}}}} \right)} \right./T$ and ${{{\tilde R}_{\rm{s}}}\left| {_{{T_1} = T_{1{\rm{s}}}^*} = {h_{\rm{s}}}\left( {T_{1{\rm{s}}}^*,{A_{{\rm{s2}}}}} \right)} \right.}/T$. Then, we calculate ${h_{\rm{c}}}\left( {T_{1{\rm{s}}}^*,{A_{{\rm{c2}}}}} \right) - {h_{\rm{s}}}\left( {T_{1{\rm{s}}}^*,{A_{{\rm{s2}}}}} \right)$ as
\begin{align}\label{rate_difference1}
&{h_{\rm{c}}}\left( {T_{1{\rm{s}}}^*,{A_{{\rm{c2}}}}} \right) - {h_{\rm{s}}}\left( {T_{1{\rm{s}}}^*,{A_{{\rm{s2}}}}} \right)\nonumber\\
& = T_{1{\rm{s}}}^*{{\bar R}_{{\rm{c}}1}} + \left( {T - T_{1{\rm{s}}}^*} \right)\left( {R\left( {T_{1{\rm{s}}}^*,{A_{{\rm{c2}}}}} \right) - R\left( {T_{1{\rm{s}}}^*,{A_{{\rm{s}}2}}} \right)} \right),
\end{align}
where $R\left( {t,A} \right) = {\log _2}\left( {1 + {\gamma _{\rm{c}}}u\left( {A,t} \right)} \right)$. By exploiting the well-known Lagrange mean value theorem, we obtain
\begin{align}\label{rate_difference_temp}
R\left( {T_{1{\rm{s}}}^*,{A_{{\rm{c2}}}}} \right) \!-\! R\left( {T_{1{\rm{s}}}^*,{A_{{\rm{s}}2}}} \right)\! =\! {\left. {\frac{{\partial R\left( {T_{1{\rm{s}}}^*,A} \right)}}{{\partial A}}} \right|_{A = \xi }}\left( {{A_{{\rm{c2}}}} \!-\! {A_{{\rm{s}}2}}} \right),
\end{align}
where $\xi$ is a constant located in $\xi  \in \left[ {{A_{{\rm{s}}2}},{A_{{\rm{c2}}}}} \right]$. It can be easily verified that $\frac{{\partial R\left( {T_{1{\rm{s}}}^*,A} \right)}}{{\partial A}} =  - \frac{{{\gamma _{\rm{c}}}{u^3}\left( {A,T_{1{\rm{s}}}^*} \right)}}{{2T_{1{\rm{s}}}^*\left( {1 + {\gamma _{\rm{c}}}u\left( {A,T_{1{\rm{s}}}^*} \right)} \right)\ln 2}}$ is a decreasing function with respect to $A$. Consequently, we have ${\left. {\frac{{\partial R\left( {T_{1{\rm{s}}}^*,A} \right)}}{{\partial A}}} \right|_{A = \xi }} \ge {\left. {\frac{{\partial R\left( {T_{1{\rm{s}}}^*,A} \right)}}{{\partial A}}} \right|_{A = {A_{{\rm{s}}2}}}}$.
Since $\frac{{\partial R\left( {T_{1{\rm{s}}}^*,A} \right)}}{{\partial A}} < 0$ and ${A_{{\rm{c2}}}} - {A_{{\rm{s}}2}} > 0$, we obtain a lower bound of \eqref{rate_difference_temp} as
\begin{align}\label{rate_difference_lowerbound}
R\left( {T_{1{\rm{s}}}^*,{A_{{\rm{c2}}}}} \right) - R\left( {T_{1{\rm{s}}}^*,{A_{{\rm{s}}2}}} \right) \ge \frac{{ - {\gamma _{\rm{c}}}{u^3}\left( {{A_{{\rm{s}}2}},T_{1{\rm{s}}}^*} \right)\left( {{A_{{\rm{c2}}}} - {A_{{\rm{s2}}}}} \right)}}{{2T_{1{\rm{s}}}^*\left( {1 + {\gamma _{\rm{c}}}u\left( {{A_{{\rm{s}}2}},T_{1{\rm{s}}}^*} \right)} \right)\ln 2}}.
\end{align}
By plugging \eqref{rate_difference_lowerbound} into \eqref{rate_difference1}, we have
\begin{align}\label{rate_difference_lowerbound1}
&{h_{\rm{c}}}\left( {T_{1{\rm{s}}}^*,{A_{{\rm{c2}}}}} \right) - {h_{\rm{s}}}\left( {T_{1{\rm{s}}}^*,{A_{{\rm{s2}}}}} \right)\nonumber\\
& \ge T_{1{\rm{s}}}^*{{\bar R}_{{\rm{c}}1}} - \left( {T - T_{1{\rm{s}}}^*} \right)\frac{{{\gamma _{\rm{c}}}{u^3}\left( {{A_{{\rm{s}}2}},T_{1{\rm{s}}}^*} \right)\left( {{A_{{\rm{c2}}}} - {A_{{\rm{s2}}}}} \right)}}{{2T_{1{\rm{s}}}^*\left( {1 + {\gamma _{\rm{c}}}u\left( {{A_{{\rm{s}}2}},T_{1{\rm{s}}}^*} \right)} \right)\ln 2}}.
\end{align}
Notice that ${T_{1{\rm{s}}}^*}$ is the root of equation \eqref{equation_s}, ${T_{1{\rm{s}}}^*}$ satisfies
\begin{align}\label{TS_condition}
\frac{{T \!-\! T_{1{\rm{s}}}^*}}{{2T_{1{\rm{s}}}^{*2}}} \!=\! \frac{{\ln \left( {1 \!+\! {\gamma _{\rm{c}}}u\left( {{A_{{\rm{s2}}}},T_{1{\rm{s}}}^*} \right)} \right)\left( {1 \!+\! {\gamma _{\rm{c}}}u\left( {{A_{{\rm{s2}}}},T_{1{\rm{s}}}^*} \right)} \right)}}{{{\gamma _{\rm{c}}}{A_{{\rm{s2}}}}{u^3}\left( {{A_{{\rm{s2}}}},T_{1{\rm{s}}}^*} \right)}}.
\end{align}
Combining the results in \eqref{rate_difference_lowerbound1} and \eqref{TS_condition}, we obtain
\begin{align}\label{rate_difference_lowerbound2}
&{h_{\rm{c}}}\left( {T_{1{\rm{s}}}^*,{A_{{\rm{c2}}}}} \right) - {h_{\rm{s}}}\left( {T_{1{\rm{s}}}^*,{A_{{\rm{s2}}}}} \right)\nonumber\\
&\ge T_{1{\rm{s}}}^*\left( {{{\bar R}_{{\rm{c}}1}} - \frac{{{{\log }_2}\left( {1 + {\gamma _{\rm{c}}}u\left( {{A_{{\rm{s2}}}},T_{1{\rm{s}}}^*} \right)} \right)\left( {{A_{{\rm{c2}}}} - {A_{{\rm{s2}}}}} \right)}}{{{A_{{\rm{s2}}}}}}} \right)\nonumber\\
&\mathop  \ge \limits^{\left( a \right)} T_{1{\rm{s}}}^*\left( {{{\bar R}_{{\rm{c}}1}} - \log \left( {1 + {\gamma _{\rm{c}}}} \right)\left( {\frac{{{A_{{\rm{c2}}}}}}{{{A_{{\rm{s2}}}}}} - 1} \right)} \right),
\end{align}
where (a) follows by $u\left( {{A_{{\rm{s2}}}},T_{1{\rm{s}}}^*} \right) < 1$.

Based on \eqref{rate_difference_lowerbound2}, it can be observed that ${h_{\rm{c}}}\left( {T_{1{\rm{s}}}^*,{A_{{\rm{c2}}}}} \right) - {h_{\rm{s}}}\left( {T_{1{\rm{s}}}^*,{A_{{\rm{s2}}}}} \right) > 0$ always holds if the condition ${{\bar R}_{{\rm{c}}1}} > {\log _2}\left( {1 + {\gamma _{\rm{c}}}} \right)\left( {\frac{{{A_{{\rm{c2}}}}}}{{{A_{{\rm{s2}}}}}} - 1} \right)$ is satisfied. Notice that $T_{1{\rm{s}}}^*$ is generally suboptimal for problem \eqref{opt_time_app1}, which leads to that $\tilde R_{\rm{c}}^* \ge {h_{\rm{c}}}\left( {T_{1{\rm{s}}}^*,{A_{{\rm{c2}}}}} \right)/T$. Hence, $\tilde R_{\rm{c}}^* \ge {h_{\rm{c}}}\left( {T_{1{\rm{s}}}^*,{A_{{\rm{c2}}}}} \right)/T > {h_{\rm{s}}}\left( {T_{1{\rm{s}}}^*,{A_{{\rm{s2}}}}} \right)/T = \tilde R_{\rm{s}}^*$ always holds under the condition that ${{\bar R}_{{\rm{c}}1}} > {\log _2}\left( {1 + {\gamma _{\rm{c}}}} \right)\left( {\frac{{{A_{{\rm{c2}}}}}}{{{A_{{\rm{s2}}}}}} - 1} \right)$. The proof is completed.

\section*{Appendix F: \textsc{Proof of Proposition 3}}
We first focus on the analysis of $T_{1{\rm{s}}}^*$ under an  asymptotically large $T$. Notice that $T_{1{\rm{s}}}^*$ satisfies equation \eqref{TS_condition}. To guarantee the right hand side of \eqref{TS_condition} to be a finite value, $T_{1{\rm{s}}}^* = k\sqrt T $ must hold, where $k$ is a constant related to the system parameters. By plugging $T_{1{\rm{s}}}^* = k\sqrt T $ into \eqref{TS_condition} and letting $T \to \infty$, we obtain
\begin{align}\label{TS_condition_a}
\frac{1}{{2{k^2}}} = \frac{{\ln \left( {1 + {\gamma _{\rm{c}}}} \right)\left( {1 + {\gamma _{\rm{c}}}} \right)}}{{{\gamma _{\rm{c}}}{A_{{\rm{s2}}}}}},
\end{align}
since ${u\left( {{A_{{\rm{s2}}}},T_{1{\rm{s}}}^*} \right)} = 1$ as $T \to \infty$. Based on \eqref{TS_condition_a}, we obtain
\begin{align}\label{constant_scale}
k = \sqrt {\frac{{{\gamma _{\rm{c}}}{A_{{\rm{s}}2}}}}{{2\ln 2\left( {1 + {\gamma _{\rm{c}}}} \right){C_0}}}} ,
\end{align}
which directly leads to the result in Proposition 3. By following the similar steps, $T_{1{\rm{c}}}^*$ in \eqref{asymptotically_optimal_solution} can be obtained, which thus completes the proof.

\section*{Appendix G: \textsc{Proof of Proposition 4}}
First, we aim to demonstrate that $\mathop {\lim }\limits_{{P_t} \to \infty } \frac{{{C_0} - \tilde R_{\rm{s}}^*}}{{\tilde R_{\rm{s}}^*}} = 0$ by showing that $\mathop {\lim }\limits_{{P_t} \to \infty } \frac{{{C_0} - \tilde R_{\rm{s}}^*}}{{\tilde R_{\rm{s}}^*}} \ge 0$ and $\mathop {\lim }\limits_{{P_t} \to \infty } \frac{{{C_0} - \tilde R_{\rm{s}}^*}}{{\tilde R_{\rm{s}}^*}} \le 0$. Notice that $\mathop {\lim }\limits_{{P_t} \to \infty } \frac{{{C_0} - \tilde R_{\rm{s}}^*}}{{\tilde R_{\rm{s}}^*}} \ge 0$ hodes naturally since we have ${C_0} \ge \tilde R_{\rm{s}}^*$. Then, we focus on proving that $\mathop {\lim }\limits_{{P_t} \to \infty } \frac{{{C_0} - \tilde R_{\rm{s}}^*}}{{\tilde R_{\rm{s}}^*}} \le 0$. By using the time allocation results in Proposition 3, we construct a lower bound of $\tilde R_{\rm{s}}^*$ as
\begin{align}\label{Rs_lb}
\tilde R_{\rm{s}}^{{\rm{lb}}} = \left( {1 - \frac{k}{{\sqrt T }}} \right){\log _2}\left( {1 + {\gamma _{\rm{c}}}{{\left( {1 + \frac{{{A_{{\rm{s}}2}}}}{{k\sqrt T }}} \right)}^{ - 1/2}}} \right),
\end{align}
where $k$ is defined in \eqref{constant_scale}. Let $\chi  = \frac{{{\pi ^2}\left( {M_r^2 - 1} \right){\sigma ^2}}}{{12{{\left\| {{{{\bf{\dot b}}}_{\rm{s}}}} \right\|}^2}}}\frac{{{M_r}}}{{{\sigma ^2}}}$ denote a constant which independent of the transmit power ${{P_t}}$. We further rewrite $\tilde R_{\rm{s}}^{{\rm{lb}}}$ as
\begin{align}\label{Rs_lb2}
\tilde R_{\rm{s}}^{{\rm{lb}}} \!=\! \left( {1 \!-\! \sqrt {\frac{{\chi {{\left( {1 + {\gamma _{\rm{c}}}} \right)}^{ - 1}}}}{{2T\ln \left( {1 \!+\! {\gamma _{\rm{c}}}} \right)}}} } \right){\log _2}\left( {1 \!\!+\!\! \frac{{{\gamma _{\rm{c}}}}}{{\sqrt {1 \!+\! g\left( {{\gamma _c}} \right)} }}} \right),
\end{align}
where $g\left( {{\gamma _c}} \right) = \sqrt {\frac{{2\left( {1 + {\gamma _{\rm{c}}}} \right)\chi \ln \left( {1 + {\gamma _{\rm{c}}}} \right)}}{{T\gamma _{\rm{c}}^2}}} $. Since ${{P_t} \to \infty }$ is equivalent to ${\gamma _{\rm{c}}} \to \infty $, we obtain $\tilde R_{\rm{s}}^{{\rm{lb}}}$ in the regime ${\gamma _{\rm{c}}} \to \infty $ as
\begin{align}\label{Rs_lb3}
\tilde R_{\rm{s}}^{{\rm{lb}}} \!=\! \frac{1}{T}\left( {T \!-\! {\cal O}\left( {\frac{1}{{{\gamma _{\rm{c}}}{{\log }_2}{\gamma _{\rm{c}}}}}} \right)} \right)\left( {{C_0} \!-\! {\cal O}\left( {\sqrt {\frac{{{{\log }_2}{\gamma _{\rm{c}}}}}{{{\gamma _{\rm{c}}}}}} } \right)} \right),
\end{align}
which leads to the result $\mathop {\lim }\limits_{{P_t} \to \infty } \frac{{\tilde R_{\rm{s}}^{{\rm{lb}}}}}{{{C_0}}} = \mathop {\lim }\limits_{{\gamma _{\rm{c}}} \to \infty } \frac{{\tilde R_{\rm{s}}^{{\rm{lb}}}}}{{{C_0}}} = 1$. Thus, we obtain  $\mathop {\lim }\limits_{{P_t} \to \infty } \frac{{{C_0} - \tilde R_{\rm{s}}^*}}{{\tilde R_{\rm{s}}^*}} \le 0$ since $\tilde R_{\rm{s}}^{\rm{*}} \ge \tilde R_{\rm{s}}^{{\rm{lb}}}$ holds. By combining the results of $\mathop {\lim }\limits_{{P_t} \to \infty } \frac{{{C_0} - \tilde R_{\rm{s}}^*}}{{\tilde R_{\rm{s}}^*}} \ge 0$ and $\mathop {\lim }\limits_{{P_t} \to \infty } \frac{{{C_0} - \tilde R_{\rm{s}}^*}}{{\tilde R_{\rm{s}}^*}} \le 0$, we obtain $\mathop {\lim }\limits_{{P_t} \to \infty } \frac{{{C_0} - \tilde R_{\rm{s}}^*}}{{\tilde R_{\rm{s}}^*}} = 0$. By following the similar steps, $\mathop {\lim }\limits_{{P_t} \to \infty } \frac{{{C_0} - \tilde R_{\rm{c}}^*}}{{\tilde R_{\rm{c}}^*}} = 0$ can be also obtained and we omit the details due to its brevity, which thus completes the proof.

\bibliographystyle{IEEEtran}
\vspace{-8pt}
\bibliography{IEEEabrv,myref}


\end{document}